\documentclass[11pt]{article}
\usepackage{hyperref}
\usepackage{times}  
\usepackage{mathpazo}
\usepackage{amssymb,amsmath,amsthm}
\usepackage{epsfig}
\usepackage{tcolorbox}
\usepackage{lipsum} 
\usepackage{enumitem}
\usepackage{xspace}
\usepackage{tikz}
\usepackage[T1]{fontenc}

\tcbset{ 
  halign=justify, 
  center, 
  colback=gray!10, 
  colframe=black, 
  boxrule=0.5pt 
}

\newtheorem{theorem}{Theorem}[section]
\newtheorem{proposition}[theorem]{Proposition}
\newtheorem{definition}[theorem]{Definition}

\newtheorem{lemma}[theorem]{Lemma}

\newtheorem{corollary}[theorem]{Corollary}

\newtheorem{remark}[theorem]{Remark}

\def\lsoft{{l\kern-0.035cm\char39\kern-0.03truecm}}

\newcommand{\qedsymb}{\hfill{\rule{2mm}{2mm}}}
\renewenvironment{proof}[1][]{\begin{trivlist}
\item[\hspace{\labelsep}{\bf\noindent Proof#1:\/}] }{\qedsymb\end{trivlist}}

\def\calG{{\cal G}}
\def\calF{{\cal F}}

\def\calP{{\cal P}}

\def\R{\mathbb{R}}

\def\N{\mathbb{N}}

\newcommand{\true}{\mathsf{True}}
\newcommand{\false}{\mathsf{False}}

\newcommand{\NP}{\mathsf{NP}}

\newcommand{\coNPpoly}{\mathsf{coNP/poly}}

\newcommand{\YES}{\mathsf{YES}}
\newcommand{\NO}{\mathsf{NO}}

\newcommand{\eps}{\epsilon}
\renewcommand{\epsilon}{\varepsilon}

\newcommand{\Fset}{\mathbb{F}}         

\newcommand{\URFC}{\textsc{Uniformly Rainbow Free Coloring}\xspace}
\newcommand{\GURFC}{\textsc{Generalized Uniformly Rainbow Free Coloring}\xspace}

\newcommand{\ListCol}{\textsc{List Coloring}\xspace}
\newcommand{\ListHCol}{\textsc{List $H$-Coloring}\xspace}
\newcommand{\CLC}{\textsc{Constrained List Coloring}\xspace}
\newcommand{\CC}{\textsc{Constrained Coloring}\xspace}
\newcommand{\Col}{\textsc{Coloring}\xspace}
\newcommand{\lClique}{\textsc{-DisjCliques}\xspace}
\newcommand{\Clique}{\textsc{DisjCliques}\xspace}
\newcommand{\SAT}{\textsc{SAT}\xspace}
\newcommand{\NAE}{\textsc{NAE}\xspace}
\newcommand{\UR}{\textsc{UR}\xspace}
\newcommand{\NUR}{\textsc{URF}\xspace}

\newcommand{\OR}{\textsc{or}\xspace}

\begin{document}

\title{{\bf Kernelization Bounds for Constrained Coloring}}

\author{
Ishay Haviv\thanks{The Academic College of Tel Aviv-Yaffo, Tel Aviv, Israel.}
}

\date{}

\maketitle

\begin{abstract}
We study the kernel complexity of constraint satisfaction problems over a finite domain, parameterized by the number of variables, whose constraint language consists of two relations: the non-equality relation and an additional permutation-invariant relation $R$. We establish a conditional lower bound on the kernel size in terms of the largest arity of an \OR relation definable from $R$. Building on this, we investigate the kernel complexity of uniformly rainbow free coloring problems. In these problems, for fixed positive integers $d$, $\ell$, and $q \geq d$, we are given a graph $G$ on $n$ vertices and a collection $\calF$ of $\ell$-tuples of $d$-subsets of its vertex set, and the goal is to decide whether there exists a proper coloring of $G$ with $q$ colors such that no $\ell$-tuple in $\calF$ is uniformly rainbow, that is, no tuple has all its sets colored with the same $d$ distinct colors. We determine, for all admissible values of $d$, $\ell$, and $q$, the infimum over all values $\eta$ for which the problem admits a kernel of size $O(n^\eta)$, under the assumption \mbox{$\NP \nsubseteq \coNPpoly$}. As applications, we obtain nearly tight bounds on the kernel complexity of various coloring problems under diverse settings and parameterizations. This includes graph coloring problems parameterized by the vertex-deletion distance to a disjoint union of cliques, resolving a question of Schalken (2020), as well as uniform hypergraph coloring problems parameterized by the number of vertices, extending results of Jansen and Pieterse (2019) and Beukers (2021).
\end{abstract}


\section{Introduction}

Constraint satisfaction problems form a central concept in theoretical computer science, interacting with foundational areas in the field, such as approximation algorithms, hardness of approximation, fine-grained complexity, parameterized complexity, and kernelization. In the {\em constraint satisfaction problem} associated with a relation $R \subseteq D^r$ of arity $r$ over a finite domain $D$, the input consists of a set of variables $V$ and a collection $\calF \subseteq V^r$ of $r$-tuples of variables, called {\em constraints}, and the goal is to decide whether there exists an assignment $c: V \to D$ that satisfies all constraints in $\calF$, meaning that for every $r$-tuple $(x_1, \ldots, x_r) \in \calF$, it holds that \mbox{$(c(x_1), \ldots, c(x_r)) \in R$.} More broadly, a constraint satisfaction problem can be defined with respect to a set of relations over $D$, called the {\em constraint language}, where each constraint is associated with one of them. This framework naturally encompasses central decision problems studied in theoretical computer science, such as bounded-width satisfiability, graph and hypergraph coloring, and many others. A remarkable result in this context is the dichotomy theorem, which asserts that every constraint satisfaction problem is either solvable in polynomial time or $\NP$-hard. This was proved for the Boolean domain in 1978 by Schaefer~\cite{Schaefer78}, and more recently, for all finite domains, by Bulatov~\cite{Bulatov17} and Zhuk~\cite{Zhuk20} independently.

The present paper studies constraint satisfaction problems from the perspective of {\em kernelization} (see, e.g.,~\cite{KernelBook19}). This theme, situated in the area of parameterized complexity, seeks to determine to what extent instances of a computationally hard problem can be compressed with respect to a chosen parameter of the instances. Specifically, a kernelization algorithm for a parameterized problem is given as input a pair $(x,k)$, where $x$ is the main input and $k$ is the associated parameter, and produces in polynomial time an equivalent instance, whose size is bounded by a function that depends solely on $k$. A major challenge in the field is to estimate, under plausible complexity assumptions, the smallest possible growth rate of the kernel size of a given parameterized problem. When parameterized by the number of variables $n$, a constraint satisfaction problem associated with a constraint language of relations of arity at most $r$ admits a trivial kernelization algorithm that produces instances with $O(n^r)$ constraints, simply by removing duplicate constraints. Determining the smallest possible kernel size for such problems has attracted significant attention in recent years, e.g.,~\cite{DellM14,JansenP19sparse,LagerkvistW20,ChenJP20,JansenW24}, and is the driving force behind this work.

The kernel complexity of constraint satisfaction problems over the Boolean domain $D=\{0,1\}$ has been studied extensively in the literature. A remarkable result, proved by Dell and van Melkebeek~\cite{DellM14}, states that for every integer $k \geq 3$ and for any real $\eps >0$, the satisfiability problem on CNF formulas with clauses of size $k$, when parameterized by the number of variables $n$, does not admit a kernel of size $O(n^{k-\eps})$ unless \mbox{$\NP \subseteq \coNPpoly$}. Since this containment is known to imply the collapse of the polynomial-time hierarchy~\cite{Yap83}, it follows that the trivial kernel of size $O(n^k)$ is presumably near-optimal. Yet, for other variants of the satisfiability problem, smaller kernels are known to exist. For example, Jansen and Pieterse~\cite{JansenP17,JansenP19sparse} showed that the Not All Equal satisfiability problem on CNF formulas with clauses of size $k \geq 3$, when parameterized by the number of variables $n$, admits a kernel with $O(n^{k-1})$ clauses, which can be encoded in $O(n^{k-1} \cdot \log n)$ bits. This was achieved via an algebraic approach they introduced in~\cite{JansenP17} based on low-degree polynomials expressing the satisfiability constraint. They further provided a nearly matching lower bound for the problem, ruling out any kernel of size $O(n^{k-1-\eps})$ for a constant $\eps>0$ unless $\NP \subseteq \coNPpoly$ (see also~\cite{ChenJP20}). More generally, Chen, Jansen, and Pieterse~\cite{ChenJP20} provided a criterion for when an intractable Boolean constraint satisfaction problem admits a non-trivial kernel. Letting $r$ denote the maximum arity of the relations in the constraint language at hand, they showed that if an \OR relation of arity $r$ is definable from one of them, in the sense that it has exactly one falsifying assignment, then no kernel of size $O(n^{r-\eps})$ exists for $\eps>0$ unless $\NP \subseteq \coNPpoly$, whereas otherwise, the number of constraints can be efficiently reduced to $O(n^{r-1})$. This result was used in~\cite{ChenJP20} to obtain an exact conditional characterization of the kernel complexity of all Boolean constraint satisfaction problems with relations of arity at most three. The kernel complexity of constraint satisfaction problems over finite domains of size larger than $2$ is currently much less understood. To address these cases, Lagerkvist and Wahlstr{\"o}m~\cite{LagerkvistW20} generalized the approach of~\cite{ChenJP20} using tools from universal algebra. Among their results lies a full conditional characterization of the finite-domain constraint satisfaction problems that admit a kernel whose number of constraints grows linearly with the number of variables.

The study of the kernel complexity of non-Boolean constraint satisfaction problems is strongly motivated by graph coloring problems, which provide a rich source of kernelization challenges. For an integer $q$, the $q$-\Col problem asks whether a given graph admits a {\em proper} coloring with $q$ colors, i.e., an assignment of colors to the vertices such that adjacent vertices receive distinct colors. This problem, which is well known to be $\NP$-hard for every $q \geq 3$, corresponds to the constraint satisfaction problem associated with the {\em non-equality} relation of arity $2$ over the domain $[q]= \{1,\ldots,q\}$. When parameterized by the number of vertices $n$, the problem admits a trivial kernel of size $O(n^2)$, and a result of Jansen and Pieterse~\cite{JansenP17,JansenP19color} asserts that, for every $q \geq 3$ and for any $\eps >0$, it admits no kernel of size $O(n^{2-\eps})$ unless $\NP \subseteq \coNPpoly$. This result has inspired natural extensions and generalizations. For example, such a lower bound was obtained for certain (list) $H$-\Col problems, which ask whether an input graph admits a homomorphism to a fixed target graph $H$, possibly subject to lists of allowed vertices in $H$ for each vertex~\cite{ChenJOPR23,BerkmanH25}. Another extension, due to Beukers~\cite{Beukers21}, concerns the $q$-\Col problem on $3$-uniform hypergraphs with $n$ vertices, corresponding to the relation of arity $3$ that consists of all triples in $[q]^3$ except those with all equal entries. Under the assumption \mbox{$\NP \nsubseteq \coNPpoly$}, he proved that for every $q \geq 3$ and for any $\eps >0$, the problem admits no kernel of size $O(n^{3-\eps})$, essentially matching the trivial upper bound.

It is noteworthy that recent research on the kernel complexity of constraint satisfaction problems has revealed intimate connections to additional structural quantities of interest. A prominent example is the {\em non-redundancy} of such problems, studied, e.g., in~\cite{BessiereCK20,Carbonnel22,BrakensiekG25,BGJLW25}, which measures the maximum number of constraints in a non-redundant instance on $n$ variables, that is, an instance in which no constraint is implied by the others. At first glance, this quantity may appear unrelated to the optimal kernel size, since a kernelization algorithm is not restricted to removing redundant constraints and may modify the instance in arbitrary ways, even introducing new constraints. Furthermore, the kernelization task is algorithmic, so tractable constraint satisfaction problems admit a kernel of constant size, whereas their non-redundancy can grow polynomially. Also, while combining several tractable relations in the constraint language of a constraint satisfaction problem can render the problem intractable, and thus result in non-constant kernel complexity, the non-redundancy of the problem is dominated by the maximal non-redundancy of the individual relations. Nevertheless, recent works indicate that in many cases, the two quantities are closely tied. Known kernelization algorithms typically operate by eliminating redundant constraints, and upper bounds on non-redundancy can often be transformed into efficient kernelization algorithms. Clarifying the precise relationship between these two measures of constraint satisfaction problems is a compelling question in the area.

\subsection{Our Contribution}

This paper addresses the kernel complexity of constraint satisfaction problems over a finite domain of arbitrary size $q$, whose constraint language consists of two relations: the non-equality relation of arity $2$ over $[q]$, corresponding to the $q$-\Col problem, and an additional relation $R \subseteq [q]^r$ of arity $r$. Our first contribution establishes a general conditional lower bound on the kernel complexity of such problems. Building on this result, we investigate a family of problems, which we refer to as uniformly rainbow free coloring, and obtain nearly matching upper and lower bounds on their kernel complexity. We then use these bounds to determine the kernel complexity of various coloring problems under diverse settings and parameterizations. This includes graph coloring problems parameterized by the vertex-deletion distance to a disjoint union of cliques, as well as uniform hypergraph coloring problems parameterized by the number of vertices. A detailed description of our contribution is given below.

\subsubsection{Constrained Coloring}

The constrained coloring problem associated with a relation $R$ is defined as follows.

\begin{definition}\label{def:R-CC}
For positive integers $q$ and $r$, let $R \subseteq [q]^r$ be a relation of arity $r$ over $[q]$.
The input of the $R$-\CC problem consists of a graph $G=(V,E)$ and a collection $\calF \subseteq V^r$ of $r$-tuples of vertices.
The goal is to decide whether there exists a proper coloring $c:V \to [q]$ of $G$, such that for every $r$-tuple $(x_1, \ldots, x_r) \in \calF$, it holds that $(c(x_1), \ldots, c(x_r)) \in R$. The number of constraints in the instance $(G,\calF)$ is $|E|+|\calF|$.
\end{definition}

Our first result establishes a lower bound on the kernel complexity of $R$-\CC problems for permutation-invariant relations $R$ with respect to the number of vertices parameterization. A relation $R \subseteq [q]^r$ is called {\em permutation-invariant} if every permutation of $[q]$ applied to the entries of a tuple does not affect its membership in $R$ (see Definition~\ref{def:perm}). For an integer $k$, we say that {\em an \OR relation of arity $k$ is definable from $R \subseteq [q]^r$} if there exist $r$ subsets $D_1, \ldots, D_r$ of $[q]$, with $k$ of them of size $2$ and the others of size $1$, so that among the $2^k$ tuples in the product $D_1 \times \cdots \times D_r$, exactly one does not belong to $R$ (see Definition~\ref{def:OR definable}). We prove that for every permutation-invariant relation $R$ that defines an \OR relation of arity at least $3$, the $R$-\CC problem parameterized by the number of vertices is unlikely to admit a kernel whose size is bounded by a polynomial of degree smaller than that arity. In fact, here and throughout, the lower bounds apply not only to kernels but also to {\em compressions}, i.e., algorithms that reduce an instance of the problem at hand to an instance of {\em any} problem, not necessarily itself (see Section~\ref{sec:kernel}).

\begin{theorem}\label{thm:Lower-R-CC}
For positive integers $q$, $r$, and $k \geq 3$, let $R \subseteq [q]^r$ be a permutation-invariant relation of arity $r$ over $[q]$ from which an \OR relation of arity $k$ is definable. Then, for any real $\eps >0$, the $R$-\CC problem parameterized by the number of vertices $n$ does not admit a compression of size $O(n^{k-\eps})$ unless $\NP \subseteq \coNPpoly$.
\end{theorem}

Several remarks are in order here.
First, for every relation $R \subseteq [q]^r$ of arity $r \geq 3$, the $R$-\CC problem parameterized by the number of vertices $n$ admits a trivial kernel of size $O(n^r)$, obtained by eliminating repeated constraints. Theorem~\ref{thm:Lower-R-CC} shows that if an \OR relation of arity $r$ is definable from $R$, then this trivial kernel is likely near-optimal. It turns out that this condition is essential for such a lower bound, as a result of Carbonnel~\cite{Carbonnel22} implies that if the largest arity of an \OR relation definable from $R$ is smaller than $r$, then the problem admits a smaller kernel with $O(n^{r-\eps})$ constraints for $\eps = 2^{1-r}$ (see Proposition~\ref{prop:Carbonnel}).
Second, if one is interested in a lower bound on the compression size of a {\em list-coloring} variant of the $R$-\CC problem, where the input specifies a set of available colors for each vertex, then the assumption of Theorem~\ref{thm:Lower-R-CC} that the relation $R$ is permutation-invariant can be omitted. In the sequel, we first prove a lower bound for this list-coloring variant for an arbitrary relation $R$ and then apply it to derive Theorem~\ref{thm:Lower-R-CC} (see Theorem~\ref{thm:Lower-R-CLC} and Lemma~\ref{lemma:non-listCC}).
Finally, we stress that the lower bound on the compression size of the $R$-\CC problem, given in  Theorem~\ref{thm:Lower-R-CC}, is governed by the arity of the \OR relations definable from the relation $R$, even when the constraint satisfaction problem involving only $R$ is solvable in polynomial time. This outcome arises from the interplay between the $R$-constraints and the non-equality constraints imposed by the proper coloring requirement. Moreover, the theorem sometimes yields a non-trivial lower bound even when the two individual components of the $R$-\CC problem are polynomial-time tractable, as may occur when $q=2$.

Before turning to our main applications of Theorem~\ref{thm:Lower-R-CC}, let us demonstrate its utility with a quick example concerning the $q$-\Col problem on $\ell$-uniform hypergraphs for fixed integers $\ell \geq 2$ and $q \geq 3$. In this problem, the input is a hypergraph on $n$ vertices with edges of size exactly~$\ell$, and the goal is to decide whether it admits a coloring with $q$ colors so that no edge is monochromatic. As alluded to earlier, for $\ell \in \{2,3\}$, for every $q \geq 3$, and for any $\eps>0$, the problem does not admit a compression of size $O(n^{\ell-\eps})$ unless $\NP \subseteq \coNPpoly$~\cite{JansenP19color,Beukers21}, nearly matching the trivial upper bound of $O(n^{\ell})$ on the kernel complexity. Theorem~\ref{thm:Lower-R-CC} enables us to extend the lower bound to all integers $\ell \geq 4$. To do so, consider the relation $R$ that consists of all $\ell$-tuples in $[q]^\ell$ whose entries are not all equal, and notice that the corresponding $R$-\CC problem coincides with the $q$-\Col problem on hypergraphs with edges of size $2$ and $\ell$. Observe that an \OR relation of arity $\ell$ is definable from $R$, as witnessed by restricting it to the product $\{1,2\} \times \{1,3\}^{\ell-1}$, where all $2^\ell$ tuples belong to $R$ except the all-$1$ tuple. Since $R$ is permutation-invariant, Theorem~\ref{thm:Lower-R-CC} implies that the $R$-\CC problem parameterized by the number of vertices $n$ does not admit a compression of size $O(n^{\ell-\eps})$ unless $\NP \subseteq \coNPpoly$. By a simple reduction to the $q$-\Col problem on $\ell$-uniform hypergraphs, without edges of size $2$, this yields a near-optimal lower bound on its compressibility, as summarized by the following statement.

\begin{theorem}\label{thm:Hyper}
For all integers $\ell \geq 2$ and $q \geq 3$, and for any real $\eps >0$, the $q$-\Col problem on $\ell$-uniform hypergraphs parameterized by the number of vertices $n$ does not admit a compression of size $O(n^{\ell-\eps})$ unless $\NP \subseteq \coNPpoly$.
\end{theorem}

It should be mentioned that the kernel complexity of the $q$-\Col problem on $\ell$-uniform hypergraphs parameterized by the number of vertices $n$ behaves differently when the number of colors $q$ is $2$. Indeed, it was shown in~\cite{JansenP17} that the $\NP$-hard cases of the $2$-\Col problem on $\ell$-uniform hypergraphs, with $\ell \geq 3$, admit a kernel with $O(n^{\ell-1})$ edges, which can be encoded in $O(n^{\ell-1} \cdot \log n)$ bits. We note that Theorem~\ref{thm:Hyper} can also be derived from our analysis of the coloring problems described next (see Section~\ref{sec:hyper}).

\subsubsection{Uniformly Rainbow Free Coloring}

The primary family of $R$-\CC problems considered in this paper is that of uniformly rainbow free colorings. For a graph $G=(V,E)$, a positive integer $q$, and a coloring $c:V \to [q]$, a set $A \subseteq V$ is called {\em rainbow} with respect to $c$ if $c$ assigns to the vertices of $A$ pairwise distinct colors. For positive integers $d$ and $\ell$, an $\ell$-tuple $(F_1, \ldots, F_\ell)$ of $d$-subsets of $V$ is called {\em uniformly rainbow} with respect to the coloring $c$ if all of its sets are rainbow and share the same set of colors under $c$. We define a family of computational problems as follows.

\begin{definition}\label{def:URFC}
Let $d$, $\ell$, and $q \geq d$ be positive integers.
The input of the $(d,\ell,q)$-\URFC problem consists of a graph $G=(V,E)$ and a collection $\calF$ of $\ell$-tuples of $d$-subsets of $V$.
The goal is to decide whether there exists a proper coloring $c:V \to [q]$ of $G$, such that no tuple in $\calF$ is uniformly rainbow with respect to $c$.
\end{definition}
\noindent
Note that the $(d,\ell,q)$-\URFC problem is defined for $q \geq d$, as otherwise no set of size $d$ can be rainbow, and all constraints in $\calF$ are trivially satisfied.
Note further that for any values of $d$, $\ell$, and $q$, the problem can be formulated as an $R$-\CC problem for a suitably defined permutation-invariant relation $R$ (see Definition~\ref{def:NUR}).

We provide nearly matching upper and lower bounds on the kernel complexity of the $(d,\ell,q)$-\URFC problem, parameterized by the number of vertices, for all admissible values of $d$, $\ell$, and $q$, under the complexity assumption \mbox{$\NP \nsubseteq \coNPpoly$}. We first introduce the following notation.

\begin{definition}\label{def:eta}
For positive integers $d$, $\ell$, and $q \geq d$, set
\[
\eta(d,\ell,q) =
\begin{cases}
d\ell, & \text{if } \ell \geq 2 \text{ and } q \geq d + 2, \\
d\ell - 1, & \text{if } \ell \geq 2,~ q=d+1, \text{ and } (d,\ell) \neq (1,2), \\
(d - 1)\ell, & \text{if } (\ell \geq 2 \text{ and } q = d \geq 2) \text{ or } (\ell = 1 \text{ and } d \geq 3), \\
2, & \text{if } \ell = 1,~d \leq 2, \text{ and } q \geq 3, \\
0, & \text{if } (q=2 \text{ and } d\ell \leq 2) \text{ or } q=1.
\end{cases}
\]
\end{definition}
\noindent
Note that the five cases in Definition~\ref{def:eta} are disjoint and cover all triples $(d,\ell,q)$ with $q \geq d$.
With this definition in place, we proceed to the formal statement.

\begin{theorem}\label{thm:Full-URFC}
For positive integers $d$, $\ell$, and $q \geq d$, set $\eta = \eta(d,\ell,q)$ as in Definition~\ref{def:eta}.
The $(d,\ell,q)$-\URFC problem parameterized by the number of vertices $n$ admits a kernel with $O(n^{\eta})$ constraints, while for any real $\eps>0$, it does not admit a compression of size $O(n^{\eta-\eps})$ unless $\NP \subseteq \coNPpoly$.
\end{theorem}

Let us unpack Theorem~\ref{thm:Full-URFC} in light of the cases arising from Definition~\ref{def:eta}. When $\ell \geq 2$ and $q \geq d+2$, the theorem shows that the $(d,\ell,q)$-\URFC problem parameterized by the number of vertices $n$ is unlikely to admit a kernel much smaller than the trivial one, whose number of constraints is bounded by $O(n^{d\ell})$. In contrast, when $q$ is either $d$ or $d+1$, a polynomially smaller kernel is achievable, apart from some exceptional cases, including those with $\eta(d,\ell,q)=0$, for which the problem is solvable in polynomial time. Note that, for $d \geq 2$, any coloring that assigns a fixed color to all vertices trivially satisfies all uniformly rainbow free constraints. Nonetheless, when combined with the proper coloring requirement of the graph, these constraints essentially govern the kernel complexity of the problem. It is also worth noting that the bounds on the number of constraints in the kernelized instances in Theorem~\ref{thm:Full-URFC} reflect their encoding size in bits, up to a logarithmic factor (see Remark~\ref{remark:encoding}).

The lower bounds in Theorem~\ref{thm:Full-URFC} are derived from Theorem~\ref{thm:Lower-R-CC} through an analysis of the \OR relations definable from the relations that characterize the $(d,\ell,q)$-\URFC problems (see Lemma~\ref{lemma:NURdefinable}). To establish the non-trivial upper bounds stated in Theorem~\ref{thm:Full-URFC}, we leverage the linear-algebraic approach of~\cite{JansenP17}. To adapt this machinery to our setting, we construct low-degree polynomials that capture the uniformly rainbow property, extending ideas from~\cite{JansenP19color,HavivR24,HavivR25} (see Definition~\ref{def:URproperty} and Lemma~\ref{lemma:C,p}). We also provide an extension of the upper bounds from Theorem~\ref{thm:Full-URFC}, needed for some of our applications, where the values of $d$ and $\ell$ are not unique but are taken from a prescribed set of possibilities. For this setting, we show that the degree of the polynomial kernel complexity is dominated by the maximum value of $\eta(d,\ell,q)$ over all allowed pairs $(d,\ell)$ (see Definition~\ref{def:GURFC} and Theorem~\ref{thm:GURFC}).

\subsubsection{Coloring Graphs Close to Disjoint Union of Cliques}

Our results on the uniformly rainbow free coloring problems enable us to settle the kernel complexity of a family of graph coloring problems parameterized by the vertex-deletion distance to a disjoint union of cliques. Following a terminology of Cai~\cite{Cai03}, for a positive integer $q$ and for a graph family $\calG$, let $q$-\Col on $\calG+k\mathrm{v}$ graphs be the problem parameterized by $k$, defined as follows. The input comprises a graph $G=(V,E)$ and a set $X \subseteq V$ of $k$ vertices, whose removal from $G$ gives a graph that belongs to $\calG$, and the goal is to decide whether $G$ admits a proper coloring with $q$ colors. We investigate the kernel complexity of the problem in its $\NP$-hard cases, with $q \geq 3$, for graph families $\calG$ of disjoint unions of cliques. Specifically, for a positive integer $t$, let $t \lClique$ denote the family of all graphs that form a disjoint union of cliques, each of size at most $t$, and let $\Clique$ denote the family of all graphs that form a disjoint union of cliques of any size.

The kernel complexity of the $q$-\Col problem on $t \lClique+k\mathrm{v}$ graphs was first studied in 2013 by Jansen and Kratsch~\cite{JansenK13} for the special case of $t=1$, corresponding to the well-studied vertex cover number parameterization. They showed that in this case, for every integer $q \geq 3$, the problem admits a kernel with $O(k^q)$ vertices and bit-size, while for any $\eps > 0$, it admits no kernel of size $O(k^{q-1-\eps})$ unless $\NP \subseteq \coNPpoly$. In 2017, the upper bound was tightened by Jansen and Pieterse~\cite{JansenP19color}, using their linear-algebraic approach~\cite{JansenP17}, to a kernel with $O(k^{q-1})$ vertices and bit-size $O(k^{q-1} \cdot \log k)$, matching the lower bound of~\cite{JansenK13} up to a $k^{o(1)}$ multiplicative term. In 2020, Schalken~\cite{Schalken20} studied the kernel complexity of the $q$-\Col problem with respect to smaller parameters, with particular attention to the $q$-\Col problem on $t \lClique+k\mathrm{v}$ graphs with $t=2$. He proved that the problem admits a kernel with $O(k^{2q-2})$ vertices and bit-size, while for any $\eps > 0$, it admits no kernel of size $O(k^{2q-3-\eps})$ unless $\NP \subseteq \coNPpoly$. For $q \in \{3,4\}$, he further provided a near-optimal kernel with $O(k^{2q-3})$ vertices and bit-size $O(k^{2q-3} \cdot \log k)$, which was then extended in~\cite{HavivR25} to all integers $q \geq 3$. Schalken~\cite{Schalken20} also studied the $q$-\Col problem on $\Clique+k\mathrm{v}$ graphs, which essentially corresponds to the case $t = q$. He showed that the problem admits a kernel with $O(k^r)$ vertices for $r = \max_{\ell \in [q]} (q-\ell+1)\ell = \lfloor (q+1)^2/4\rfloor$, and asked whether a smaller kernel exists and whether matching conditional lower bounds could be established.

Based on our bounds on the kernel complexity of uniformly rainbow free coloring problems, we determine the kernel complexity of the $q$-\Col problem on $t \lClique+k\mathrm{v}$ graphs for all admissible values of $q$ and $t$, up to a $k^{o(1)}$ multiplicative factor and under the assumption $\NP \nsubseteq \coNPpoly$ (see also Theorem~\ref{thm:LowerDisjClique_ell} and Remark~\ref{remark:encoding_disj}).

\begin{theorem}\label{thm:t-Cliques}
For positive integers $q \geq 3$ and $t \leq q$, set $r = r(q,t) = \max_{\ell \in [t]} {\eta(q-\ell+1,\ell,q)}$, with $\eta$ as in Definition~\ref{def:eta}.
The $q$-\Col problem on $t\lClique+k\mathrm{v}$ graphs admits a kernel with $O(k^r)$ vertices, while for any real $\eps >0$, it does not admit a compression of size $O(k^{r-\eps})$ unless $\NP \subseteq \coNPpoly$.
\end{theorem}

As a consequence of Theorem~\ref{thm:t-Cliques}, we obtain nearly tight bounds on the kernel complexity of the $q$-\Col problem on $\Clique+k\mathrm{v}$ graphs. Our findings resolve the open question raised in~\cite{Schalken20}, and show that the kernel size achieved there is essentially optimal for $q \geq 4$, but can be improved by almost a linear factor for $q=3$. For the precise statement, see Corollary~\ref{cor:q-Col-Disj}.

\subsection{Outline}

The remainder of the paper is structured as follows. In Section~\ref{sec:preliminaries}, we collect definitions and results that will be used throughout the paper. In Section~\ref{sec:R-Col}, we establish our general lower bound on the compression size of $R$-\CC problems and prove Theorem~\ref{thm:Lower-R-CC}. In Section~\ref{sec:URFClower}, we obtain our lower bounds on the compression size of the uniformly rainbow free coloring problems, and in Section~\ref{sec:URFCupper}, we obtain the nearly matching upper bounds on their kernel complexity, thereby confirming Theorem~\ref{thm:Full-URFC}. In Section~\ref{sec:applications}, we present several applications of our results, and, in particular, prove Theorems~\ref{thm:Hyper} and~\ref{thm:t-Cliques}. We end the paper in Section~\ref{sec:conclude} with some concluding remarks.

\section{Preliminaries}\label{sec:preliminaries}

For a positive integer $n$, let $[n]$ denote the set of positive integers up to $n$, and for a set $A$, let $\calP(A)$ denote its power set.

\subsection{Graphs and Hypergraphs}

All graphs considered in this paper are finite and simple, having no loops or parallel edges. For a graph $G=(V,E)$ and a vertex $v \in V$, we let $N_G(v)$ denote the set of neighbors of $v$ in $G$. For a set $X \subseteq V$, we denote by $G \setminus X$ the graph obtained from $G$ by removing the vertices of $X$ and by $G[X]$ the subgraph of $G$ induced by $X$. The set $X$ is called a {\em clique} when the graph $G[X]$ is complete. A graph is called a {\em disjoint union of cliques} if its vertex set can be partitioned into cliques with no edges connecting vertices from distinct parts. For a positive integer $t$, let $t \lClique$ denote the family of all graphs that form a disjoint union of cliques, each of size at most $t$, and let $\Clique$ denote the family of all graphs that form a disjoint union of cliques of arbitrary size. It is well known and easy to verify that a graph belongs to $\Clique$ if and only if it contains no induced path on three vertices.

For a positive integer $q$, a {\em coloring} of a graph $G=(V,E)$ with $q$ colors is a function from its vertex set to a set of size $q$. A coloring $c:V \to [q]$ of $G$ is called {\em proper} if $c(u) \neq c(v)$ for every edge $\{u,v\} \in E$. More generally, for a hypergraph $H=(V,E)$, a coloring $c: V \to [q]$ of $H$ is called {\em proper} if no edge in $H$ is monochromatic, namely, for each $e \in E$, it holds that $|\{c(v) \mid v \in e\}| \neq 1$. For a positive integer $\ell$, the hypergraph $H$ is called {\em $\ell$-uniform} if all of its edges are of size exactly~$\ell$. A set $A \subseteq V$ is called {\em rainbow} with respect to a coloring $c$ if $c$ assigns to its vertices pairwise distinct colors, that is, $|\{c(v) \mid v \in A\}|=|A|$. Furthermore, for positive integers $d$ and $\ell$, an $\ell$-tuple $(F_1, \ldots, F_\ell)$ of $d$-subsets of $V$ is called {\em uniformly rainbow} with respect to $c$ if all of its sets are rainbow and receive the same set of colors with respect to $c$, that is, the sets $\{c(v) \mid v \in F_j\}$ for $j \in [\ell]$ are all equal and of size $d$.

\subsection{Relations}

We present here two notions that will be used in our study of constraint satisfaction problems. We begin with the notion of a permutation-invariant relation.
\begin{definition}\label{def:perm}
For positive integers $q$ and $r$, let $R \subseteq [q]^r$ be a relation of arity $r$ over $[q]$.
The relation $R$ is called {\em permutation-invariant} if for every permutation $\pi:[q] \to [q]$ and for every $r$-tuple \mbox{$x=(x_1, \ldots, x_r) \in [q]^r$,} it holds that $x \in R$ if and only if $(\pi(x_1), \ldots, \pi(x_r)) \in R$.
\end{definition}
\noindent
Note that permutation-invariance concerns permutations of the domain elements rather than of the tuple positions. Consequently, permuting the entries of a tuple may not preserve membership in a permutation-invariant relation.

We next formalize when an \OR relation is definable from another relation. Related and more general notions can be found, e.g., in~\cite{Carbonnel22,LagerkvistW22}.

\begin{definition}\label{def:OR definable}
Let $q$, $r$, and $k$ be positive integers, and let $R \subseteq [q]^r$ be a relation of arity $r$ over $[q]$.
We say that {\em an \OR relation of arity $k$ is definable from $R$} if there exist sets $D_1, \ldots, D_r \subseteq [q]$, exactly $k$ of which are of size $2$ and the others of size $1$, such that all $2^k$ tuples in the product $D_1 \times \cdots \times D_r$ except one belong to $R$, equivalently, $| (D_1 \times \cdots \times D_r) \cap R| = 2^k-1$.
\end{definition}

\subsection{Parameterized Complexity}\label{sec:kernel}

We now gather fundamental definitions on kernelization from the field of parameterized complexity. For an in-depth introduction to the area, the reader is referred to~\cite{KernelBook19}.
A {\em parameterized problem} is a set $Q \subseteq \Sigma^* \times \N$ for some finite alphabet $\Sigma$. A {\em compression} (also known as a generalized kernel) for a parameterized problem $Q \subseteq \Sigma^* \times \N$ into a parameterized problem $Q' \subseteq \Sigma^* \times \N$ is an algorithm that given an instance $(x,k) \in \Sigma^* \times \N$ returns in time polynomial in $|x|+k$ an instance $(x',k') \in \Sigma^* \times \N$, such that $(x,k) \in Q$ if and only if $(x',k') \in Q'$, and such that $|x'|+k' \leq h(k)$ for some computable function $h:\N \to \R$. The function $h$ is called the {\em size} of the compression, and when the alphabet $\Sigma$ is $\{0,1\}$, it is called the {\em bit-size} of the compression. A parameterized problem $Q$ is said to admit a compression of size $h$ if there exists a compression of size $h$ for $Q$ into {\em some} parameterized problem. A compression for a parameterized problem $Q$ into itself is called a {\em kernelization}, or simply a {\em kernel}, for $Q$.

A {\em transformation} from a parameterized problem $Q \subseteq \Sigma^* \times \N$ into a parameterized problem $Q' \subseteq \Sigma^* \times \N$ is an algorithm that given an instance $(x,k) \in \Sigma^* \times \N$ returns in time polynomial in $|x|+k$ an instance $(x',k') \in \Sigma^* \times \N$, such that $(x,k) \in Q$ if and only if $(x',k') \in Q'$, and such that $k' \leq h(k)$ for some computable function $h:\N \to \R$. Note that in contrast to the notion of compression, the bound $h(k)$ applies here only to $k'$, not to $|x'|+k'$. If $h$ is linear, the transformation is called {\em linear-parameter}. Such transformations are useful for deriving lower bounds on the compression size of parameterized problems from those of others, as illustrated by the following standard lemma.

\begin{lemma}\label{lemma:composition}
Let $Q$ and $Q'$ be parameterized problems, whose parameters are denoted by $k$, and suppose that there exists a linear-parameter transformation from $Q$ into $Q'$.
Then, for every positive integer $d$, if $Q'$ admits a compression of size $O(k^d)$, then so does $Q$.
\end{lemma}
\begin{proof}
Composing a compression of size $O(k^d)$ for $Q'$ with a linear-parameter transformation from $Q$ into $Q'$ yields a compression of size $O(k^d)$ for $Q$.
\end{proof}

\subsection{Satisfiability Problems}

For a positive integer $k$, a $k$-CNF formula is a Boolean formula in conjunctive normal form, where each clause contains $k$ literals involving $k$ distinct variables. The classic $k$-\SAT problem asks whether a given $k$-CNF formula is satisfiable. The following theorem, proved by Dell and van Melkebeek~\cite{DellM14}, provides a conditional lower bound on the compression size of the $k$-\SAT problem with $k \geq 3$ when parameterized by the number of variables.

\begin{theorem}[\cite{DellM14}]\label{thm:SAT}
For every integer $k \geq 3$ and any real $\eps >0$, the $k$-\SAT problem parameterized by the number of variables $n$ does not admit a compression of size $O(n^{k-\eps})$ unless $\NP \subseteq \coNPpoly$.
\end{theorem}

A natural variant of $k$-\SAT is the $k$-\NAE (Not All Equal) problem, which asks whether a given $k$-CNF formula is {\em NAE-satisfiable}, that is, whether there exists an assignment to its variables such that each clause contains at least one literal evaluated to $\true$ and at least one literal evaluated to $\false$. The following theorem provides a conditional lower bound on the compression size of the $k$-\NAE problem with $k \geq 3$ when parameterized by the number of variables. The case $k \geq 4$ was shown in~\cite[Theorem~6]{JansenP17}, and the case $k=3$ follows from~\cite[Theorem~3.5]{ChenJP20}.

\begin{theorem}[{\cite{JansenP17,ChenJP20}}]\label{thm:NAE}
For every integer $k \geq 3$ and any real $\eps >0$, the $k$-\NAE problem parameterized by the number of variables $n$ does not admit a compression of size $O(n^{k-1-\eps})$ unless $\NP \subseteq \coNPpoly$.
\end{theorem}

\section{A Lower Bound for Constrained Coloring}\label{sec:R-Col}

In this section, we prove a lower bound on the compression size of $R$-\CC problems in terms of the arity of the \OR relations definable from $R$ (see Definitions~\ref{def:R-CC} and~\ref{def:OR definable}). We first introduce a list-coloring variant of the $R$-\CC problem, then prove a lower bound on its compression size, and finally derive the lower bound for the non-list variant, stated as Theorem~\ref{thm:Lower-R-CC}.

\subsection{Constrained List Coloring}

For a positive integer $q$, the input of the $q$-\ListCol problem consists of a graph $G=(V,E)$ and a function $L:V \to \calP([q])$, specifying a set of available colors for each vertex in $G$. The goal is to decide whether there exists a proper coloring $c:V \to [q]$ of $G$ satisfying $c(v) \in L(v)$ for all $v \in V$. We refer to such a coloring as a {\em proper list-coloring} of $(G,L)$. We define the $R$-\CLC problem associated with a relation $R$ as follows.

\begin{definition}\label{def:R-CLC}
For positive integers $q$ and $r$, let $R \subseteq [q]^r$ be a relation of arity $r$ over $[q]$.
The input of the $R$-\CLC problem consists of a graph $G=(V,E)$, a function \mbox{$L:V \to \calP([q])$}, and a collection $\calF \subseteq V^r$ of $r$-tuples of vertices.
The goal is to decide whether there exists a proper list-coloring $c:V \to [q]$ of $(G,L)$, such that for every $r$-tuple $(x_1, \ldots, x_r) \in \calF$, it holds that $(c(x_1), \ldots, c(x_r)) \in R$. The number of constraints in the instance $(G,L,\calF)$ is $|E|+|\calF|$.
\end{definition}

The following theorem provides a conditional lower bound on the compression size of $R$-\CLC problems parameterized by the number of vertices.

\begin{theorem}\label{thm:Lower-R-CLC}
For positive integers $q$, $r$, and $k \geq 3$, let $R \subseteq [q]^r$ be a relation of arity $r$ over $[q]$ from which an \OR relation of arity $k$ is definable. Then, for any real $\eps >0$, the $R$-\CLC problem parameterized by the number of vertices $n$ does not admit a compression of size $O(n^{k-\eps})$ unless $\NP \subseteq \coNPpoly$.
\end{theorem}

The proof of Theorem~\ref{thm:Lower-R-CLC} uses a simple graph gadget, designed to forbid a specific color assignment on two given vertices of a graph, as described in the following lemma.

\begin{lemma}\label{lemma:gadget}
For every integer $q \geq 3$, there exists a polynomial-time algorithm that, given a graph \mbox{$G=(V,E)$}, a list function $L:V \to \calP([q])$, two distinct vertices $u_1,u_2 \in V$, and colors $\alpha_1,\alpha_2 \in [q]$, extends the instance $(G,L)$ of $q$-\ListCol to an instance $(G',L')$ by adding a constant number of vertices and edges, so that for every proper list-coloring $c:V \to [q]$ of $(G,L)$, $c$ can be extended to a proper list-coloring of $(G',L')$ if and only if $(c(u_1),c(u_2)) \neq (\alpha_1,\alpha_2)$.
\end{lemma}

\begin{proof}
Fix an integer $q \geq 3$. Given a graph $G=(V,E)$, a list function $L:V \to \calP([q])$, two distinct vertices $u_1,u_2 \in V$, and colors $\alpha_1,\alpha_2 \in [q]$, we construct an instance $(G',L')$ of $q$-\ListCol as follows.
Suppose first that $\alpha_1 \neq \alpha_2$, and let $\beta$ be an element of $[q] \setminus \{\alpha_1,\alpha_2\}$.
Let $G'$ be the graph obtained from $G$ by connecting $u_1$ and $u_2$ by a path of length $3$, with two new vertices $v_1$ and $v_2$ in this order. Let $L'$ denote the list function of $G'$ that extends $L$ by $L'(v_1) = \{\alpha_1,\beta\}$ and $L'(v_2) = \{\alpha_2,\beta\}$. For an illustration, see the top path in Figure~\ref{fig:paths}.

The instance $(G',L')$ can clearly be constructed in polynomial time. We claim that it satisfies the assertion of the lemma. Let $c:V \to [q]$ be a proper list-coloring of $(G,L)$. If $c$ can be extended to a proper list-coloring of $(G',L')$, then this extension assigns the color $\beta$ to at most one of the adjacent vertices $v_1$ and $v_2$, so for some $i \in [2]$, it assigns the color $\alpha_i$ to the vertex $v_i$. Since $v_i$ is adjacent to $u_i$, this implies that $c(u_i) \neq \alpha_i$, hence $(c(u_1),c(u_2)) \neq (\alpha_1,\alpha_2)$. Conversely, if $(c(u_1),c(u_2)) \neq (\alpha_1,\alpha_2)$, then it may be assumed without loss of generality that $c(u_1) \neq \alpha_1$. We extend the coloring $c$ of $G$ to $G'$ by $c(v_1) = \alpha_1$ and by setting $c(v_2)$ as an arbitrary element of $L'(v_2) \setminus \{c(u_2)\}$. The assumption $\alpha_1 \neq \alpha_2$ implies that $\alpha_1 \notin L'(v_2)$, which yields that the obtained extension forms a proper list-coloring of $(G',L')$. This completes the argument for the case $\alpha_1 \neq \alpha_2$.

\begin{figure}[htb]
\centering
\begin{tikzpicture}[scale=1, every node/.style={circle, draw, minimum size=0.7cm}]
\node[draw=none, anchor=west] (forb1) at (-4,0) {$\alpha_1 \neq \alpha_2$};
\node (u1) at (0,0) {$u_1$};
\node[fill=gray!30] (v1) at (2,0) {$v_1$};
\node[fill=gray!30] (v2) at (4,0) {$v_2$};
\node (u2) at (6,0) {$u_2$};

\draw (u1) -- (v1) -- (v2) -- (u2);

\node[draw=none, fill=none] at (2,0.8) {$\{\alpha_1,\beta\}$};
\node[draw=none, fill=none] at (4,0.8) {$\{\alpha_2,\beta\}$};

\node[draw=none, anchor=west] (forb2) at (-4,-2) {$\alpha = \alpha_1 = \alpha_2$};
\node (U1) at (0,-2) {$u_1$};
\node[fill=gray!30] (V1) at (2,-2) {$v_1$};
\node[fill=gray!30] (V2) at (4,-2) {$v_2$};
\node[fill=gray!30] (V3) at (6,-2) {$v_3$};
\node (U2) at (8,-2) {$u_2$};

\draw (U1) -- (V1) -- (V2) -- (V3) -- (U2);

\node[draw=none, fill=none] at (2,-1.2) {$\{\alpha,\beta\}$};
\node[draw=none, fill=none] at (4,-1.2) {$\{\beta,\gamma\}$};
\node[draw=none, fill=none] at (6,-1.2) {$\{\alpha,\gamma\}$};

\end{tikzpicture}
\caption{The gadgets used in Lemma~\ref{lemma:gadget} to prevent assigning the colors $\alpha_1$ and $\alpha_2$ to vertices $u_1$ and $u_2$, respectively. The top gadget corresponds to $\alpha_1 \neq \alpha_2$ and the bottom to $\alpha_1 = \alpha_2$. Newly added vertices are shaded in gray.}
\label{fig:paths}
\end{figure}
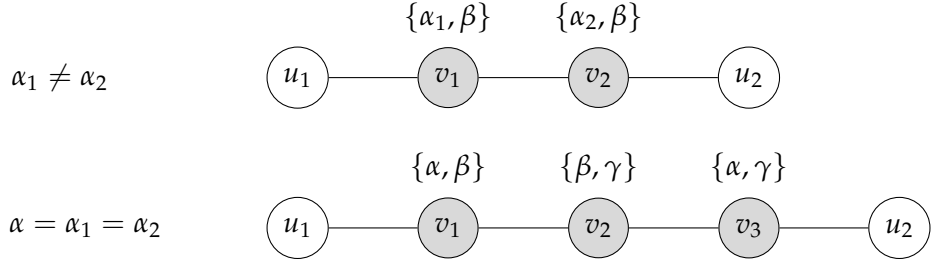

Next, suppose that $\alpha_1 = \alpha_2$, denote $\alpha = \alpha_1 = \alpha_2$, and let $\beta,\gamma$ be distinct elements of $[q] \setminus \{\alpha\}$. In this case, let $G'$ be the graph obtained from $G$ by connecting $u_1$ and $u_2$ by a path of length $4$, with three new vertices $v_1$, $v_2$, and $v_3$ in this order. Further, let $L'$ denote the list function of $G'$ that extends $L$ by $L'(v_1) = \{\alpha,\beta\}$, $L'(v_2) = \{\beta,\gamma\}$, and $L'(v_3) = \{\alpha,\gamma\}$. For an illustration, see the bottom path in Figure~\ref{fig:paths}.

As before, the instance $(G',L')$ can be constructed in polynomial time. Let $c:V \to [q]$ be a proper list-coloring of $(G,L)$. If $c$ can be extended to a proper list-coloring of $(G',L')$, then this extension cannot assign the colors $\beta$ and $\gamma$ to $v_1$ and $v_3$, respectively, because no color remains for $v_2$ in its list. This implies that at least one of $v_1$ and $v_3$ is assigned $\alpha$, hence at least one of $u_1$ and $u_2$ is not, and thus $(c(u_1),c(u_2)) \neq (\alpha,\alpha)$. Conversely, if $(c(u_1),c(u_2)) \neq (\alpha,\alpha)$, then it may be assumed without loss of generality that $c(u_1) \neq \alpha$. We extend the coloring $c$ of $G$ to $G'$ by $c(v_1) = \alpha$, $c(v_2) = \beta$, and by setting $c(v_3)$ as an arbitrary element of $L'(v_3) \setminus \{c(u_2)\}$. Since $\beta \notin L'(v_3)$, the obtained extension forms a proper list-coloring of $(G',L')$, as required.
\end{proof}

We are ready to prove Theorem~\ref{thm:Lower-R-CLC}.

\begin{proof}[ of Theorem~\ref{thm:Lower-R-CLC}]
Fix positive integers $q$, $r$, and $k \geq 3$, and a real $\eps >0 $, and let $R \subseteq [q]^r$ be a relation of arity $r$ over $[q]$ from which an \OR relation of arity $k$ is definable. By Definition~\ref{def:OR definable}, there exist a set $J = \{j_1, \ldots, j_k\} \subseteq [r]$ of $k$ distinct indices, sets $D_j = \{\alpha_j,\beta_j\} \subseteq [q]$ with $\alpha_j \neq \beta_j$ for $j \in J$, and sets $D_j = \{\alpha_j\} \subseteq [q]$ for $j \in [r] \setminus J$, such that all tuples in the product $D_1 \times \cdots \times D_r$ except the tuple $\alpha = (\alpha_1, \ldots, \alpha_r)$ belong to $R$.

Since $k \geq 3$, Theorem~\ref{thm:SAT} implies that the $k$-\SAT problem parameterized by the number of variables $n$ does not admit a compression of size $O(n^{k-\eps})$ unless $\NP \subseteq \coNPpoly$. In what follows, we present a linear-parameter transformation from the $k$-\SAT problem parameterized by the number of variables to the $R$-\CLC problem parameterized by the number of vertices. By Lemma~\ref{lemma:composition}, it follows that if the $R$-\CLC problem parameterized by the number of vertices $n$ admits a compression of size $O(n^{k-\eps})$, then so does the $k$-\SAT problem parameterized by the number of variables $n$, implying that $\NP \subseteq \coNPpoly$.

Consider an instance of the $k$-\SAT problem, namely, a $k$-CNF formula $\varphi$ on $n$ variables, denoted by $x_1, \ldots, x_n$. We transform such an instance into an instance of $R$-\CLC, consisting of a graph $G = (V,E)$, a list function $L:V \to \calP([q])$, and a collection $\calF \subseteq V^r$, as follows.
\begin{enumerate}
  \item\label{itm:t,f} The vertex set of $G$ includes $2nk$ vertices, denoted by $t_{i,j}$ and $f_{i,j}$ for all $i \in [n]$ and $j \in J$. For every pair $(i,j) \in [n] \times J$, the vertices $t_{i,j}$ and $f_{i,j}$ are adjacent in $G$, and their lists are defined by $L(t_{i,j}) = L(f_{i,j}) = D_j$. This ensures that, for every pair $(i,j) \in [n] \times J$, any proper list-coloring of $(G,L)$ assigns color $\alpha_j$ to one of $t_{i,j}$ and $f_{i,j}$ and color $\beta_j$ to the other.
  \item\label{itm:v} The graph $G$ includes $r-k$ isolated vertices, denoted by $v_{j}$ for $j \in [r] \setminus J$, whose lists are defined by $L(v_{j}) = D_j$. This ensures that any proper list-coloring of $(G,L)$ assigns color $\alpha_j$ to $v_j$ for all $j \in [r] \setminus J$.
  \item\label{itm:gadget} For each $i \in [n]$, we extend the current instance $(G,L)$ as follows. For every $s \in [k-1]$, we apply the efficient algorithm from Lemma~\ref{lemma:gadget} to the pair of vertices $(t_{i,j_s}, t_{i,j_{s+1}})$ twice: once to forbid the colors $(\alpha_{j_s}, \beta_{j_{s+1}})$ and once to forbid the colors $(\beta_{j_s}, \alpha_{j_{s+1}})$. This ensures, for each $i \in [n]$, that any proper list-coloring of $(G,L)$ assigns to all vertices $t_{i,j}$ with $j \in J$ either their corresponding colors $\alpha_j$ or their corresponding colors $\beta_j$, while imposing no further constraints on the coloring.
  \item\label{itm:F} We finally define the collection $\calF \subseteq V^r$. For each clause $(y_1 \vee \cdots \vee y_k)$ in $\varphi$, we add to $\calF$ the $r$-tuple $z=(z_1, \ldots, z_r)$, where for each $j \in [r]$, the vertex $z_j$ is defined as follows.
    \begin{itemize}
        \item If $j \in J$, let $s \in [k]$ be the index for which $j = j_s$, and let $i \in [n]$ be the index of the variable $x_i$ appearing in the literal $y_s$.
        \begin{itemize}
            \item If $y_s = x_i$, then set $z_j = t_{i,j}$.
            \item If $y_s = \overline{x_i}$, then set $z_j = f_{i,j}$.
        \end{itemize}
    \item If $j \in [r] \setminus J$, then set $z_j = v_j$.
\end{itemize}
\end{enumerate}
This completes the description of the instance $(G,L,\calF)$. It is straightforward to verify that it can be constructed in polynomial time.
We now analyze the number of vertices in $G$.
The vertex set includes the $2nk$ vertices $t_{i,j}$ and $f_{i,j}$ with $(i,j) \in [n] \times J$ from Item~\ref{itm:t,f} of the construction, as well as the $r-k$ vertices $v_j$ with $j \in [r] \setminus J$ from Item~\ref{itm:v} of the construction. Moreover, for each $i \in [n]$, Item~\ref{itm:gadget} of the construction applies the efficient algorithm from Lemma~\ref{lemma:gadget} a constant number of times to the vertices $t_{i,j}$ with $j \in J$, introducing a constant number of additional vertices.
Therefore, the total number of vertices in $G$ is $O(n)$, and the transformation is linear-parameter.

For correctness, we shall prove that the formula $\varphi$ is satisfiable if and only if $(G,L,\calF)$ is a $\YES$ instance of the $R$-\CLC problem.
Suppose first that $\varphi$ is satisfiable, and consider a satisfying assignment for $\varphi$.
We define a coloring $c:V \to [q]$ of $G$ as follows.
First, for each $j \in [r] \setminus J$, define $c(v_j) = \alpha_j$.
Next, for each $i \in [n]$, if $x_i$ is set to $\true$, then define $c(t_{i,j}) = \beta_j$ and $c(f_{i,j}) = \alpha_j$ for all $j \in J$, and if $x_i$ is set to $\false$, then define $c(t_{i,j}) = \alpha_j$ and $c(f_{i,j}) = \beta_j$ for all $j \in J$. By definition, the coloring $c$ respects the list function $L$, and for all $i \in [n]$ and $j \in J$, it assigns distinct colors to the adjacent vertices $t_{i,j}$ and $f_{i,j}$. Since for each $i \in [n]$, the vertices $t_{i,j}$ with $j \in J$ are either all assigned the corresponding colors $\alpha_j$ or all assigned the corresponding colors $\beta_j$, by Lemma~\ref{lemma:gadget}, we can properly extend $c$ to the vertices added in Item~\ref{itm:gadget} of the construction. It follows that $c$ is a proper list-coloring of $(G,L)$.

We now show that for every clause $(y_1 \vee \cdots \vee y_k)$ in $\varphi$, its associated $r$-tuple $z=(z_1, \ldots, z_r)$ in $\calF$, defined in Item~\ref{itm:F} of the construction, satisfies $c(z) = (c(z_1), \ldots, c(z_r)) \in R$. By construction, for each $j \in [r]$, we have $L(z_j) = D_j$, hence $c(z) \in D_1 \times \cdots \times D_r$. Further, for some $s \in [k]$, the literal $y_s$ is evaluated to $\true$ under the given assignment.
Let $i \in [n]$ be the index of the variable $x_i$ appearing in the literal $y_s$. If $y_s = x_i$, then the value of $x_i$ is $\true$, hence  $c(z_{j_s}) = c(t_{i,j_s}) = \beta_{j_s}$, and if $y_s = \overline{x_i}$, then the value of $x_i$ is $\false$, hence $c(z_{j_s}) = c(f_{i,j_s}) = \beta_{j_s}$. In both cases, $c(z_{j_s}) = \beta_{j_s} \neq \alpha_{j_s}$, hence $c(z) \neq \alpha$. Since the only tuple in $D_1 \times \cdots \times D_r$ that does not belong to $R$ is $\alpha$, it follows that $c(z) \in R$, as required.

For the converse direction, suppose that there exists a proper list-coloring \mbox{$c:V \to [q]$} of $(G,L)$, such that for every $r$-tuple $z=(z_1, \ldots, z_r) \in \calF$, it holds that $c(z) = (c(z_1), \ldots, c(z_r)) \in R$. We define an assignment for the variables $x_1, \ldots, x_n$ of $\varphi$ as follows. Fix some $i \in [n]$. For every $j \in J$, the vertices $t_{i,j}$ and $f_{i,j}$ are adjacent in $G$, and the list function $L$ guarantees that $c(t_{i,j}) \in D_j$ and $c(f_{i,j}) \in D_j$, where $D_j = \{\alpha_j,\beta_j\}$. Moreover, Item~\ref{itm:gadget} of the construction ensures that either $c(t_{i,j}) = \alpha_j$ for all $j \in J$, and thus $c(f_{i,j}) = \beta_j$ for all $j \in J$, or $c(t_{i,j}) = \beta_j$ for all $j \in J$, and thus $c(f_{i,j}) = \alpha_j$ for all $j \in J$. We set the variable $x_i$ to $\false$ in the former case, and to $\true$ in the latter. We claim that this assignment satisfies $\varphi$. To see this, consider a clause $(y_1 \vee \cdots \vee y_k)$ in $\varphi$ and its associated $r$-tuple $z=(z_1, \ldots, z_r)$ in $\calF$, as defined in Item~\ref{itm:F} of the construction. Notice that the list function $L$ ensures that $c(z) = (c(z_1), \ldots, c(z_r)) \in D_1 \times \cdots \times D_r$. Further, by $c(z) \in R$, it follows that $c(z) \neq \alpha$, hence there exists some $j \in J$ such that $c(z_j) = \beta_j$. Let $s \in [k]$ be the index for which $j = j_s$, and let $i \in [n]$ be the index of the variable $x_i$ that appears in the literal $y_s$. If $y_s = x_i$, then the fact that $c(t_{i,j}) = c(z_j) = \beta_j$ implies that $x_i$ is set to $\true$, and if $y_s = \overline{x_i}$, then the fact that $c(f_{i,j}) = c(z_j) = \beta_j$ implies that $x_i$ is set to $\false$. In both cases, $y_s$ is evaluated to $\true$, hence the clause is satisfied.
This completes the proof.
\end{proof}

\subsection{Constrained Coloring}

We now prove Theorem~\ref{thm:Lower-R-CC}, which asserts that if an \OR relation of arity $k \geq 3$ is definable from a permutation-invariant relation $R \subseteq [q]^r$, then the $R$-\CC problem parameterized by the number of vertices $n$ does not admit a compression of size $O(n^{k-\eps})$ for any $\eps >0$ unless $\NP \subseteq \coNPpoly$. The proof relies on the following lemma.

\begin{lemma}\label{lemma:non-listCC}
For positive integers $q$ and $r$, let $R \subseteq [q]^r$ be a permutation-invariant relation of arity $r$ over $[q]$.
Then there exists a linear-parameter transformation from the $R$-\CLC problem to the $R$-\CC problem, both parameterized by the number of vertices.
\end{lemma}

\begin{proof}
Let $R \subseteq [q]^r$ be a permutation-invariant relation.
Consider an instance of $R$-\CLC, consisting of a graph $G=(V,E)$, a list function \mbox{$L:V \to \calP([q])$}, and a collection $\calF \subseteq V^r$. Define a transformation that, given such an instance, returns the pair $(G',\calF)$, where $G'=(V',E')$ is the graph obtained from $G$ by adding a clique on $q$ vertices, denoted by $z_1, \ldots, z_q$, and connecting each vertex $v \in V$ to all vertices $z_i$ with \mbox{$i \in [q] \setminus L(v)$}. The number of vertices in $G'$ is $|V|+q$, hence the transformation is linear-parameter.

For correctness, suppose first that $(G,L,\calF)$ is a $\YES$ instance of $R$-\CLC, and consider a proper list-coloring $c: V \to [q]$ of $(G,L)$, such that for every $r$-tuple $(x_1, \ldots, x_r) \in \calF$, it holds that $(c(x_1),\ldots,c(x_r)) \in R$. Let $c'$ be the coloring of $G'$ that extends $c$ by assigning to the vertex $z_i$ the color $i$ for all $i \in [q]$. The coloring $c'$ clearly assigns distinct colors to the endpoints of every edge of $G$ and of every edge in the clique $\{z_1, \ldots, z_q\}$. Further, since $c$ respects the list function $L$, every vertex $v \in V$ satisfies $c'(v) = c(v) \in L(v)$, and thus $c'(v) \neq c'(z_i)$ whenever $i \in [q] \setminus L(v)$. This implies that $c'$ forms a proper coloring of $G'$. As an extension of $c$, it also satisfies the constraints in $\calF$, so it forms a valid solution for the instance $(G',\calF)$ of $R$-\CC.

Conversely, suppose that $(G',\calF)$ is a $\YES$ instance of $R$-\CC, and consider a proper coloring $c':V' \to [q]$ of $G'$, such that for every $r$-tuple $(x_1, \ldots, x_r) \in \calF$, it holds that $(c'(x_1),\ldots,c'(x_r)) \in R$. Since the vertices $z_1, \ldots, z_q$ are pairwise adjacent in $G'$, they receive distinct colors under $c'$, so there exists a permutation $\pi:[q] \to [q]$ such that $\pi(c'(z_i)) = i$ for all $i \in [q]$. Let $c: V \to [q]$ denote the coloring of $G$ defined by $c(v) = \pi(c'(v))$ for all $v \in V$. For every pair of adjacent vertices $u$ and $v$ in $G$, it holds that $c'(u) \neq c'(v)$, and thus $c(u) \neq c(v)$, so $c$ is a proper coloring of $G$. Furthermore, $c$ respects the list function $L$. Indeed, for every vertex $v$ and every $i \in [q] \setminus L(v)$, the vertices $v$ and $z_i$ are adjacent in $G'$, hence $c'(v) \neq c'(z_i)$, which implies that $c(v) = \pi(c'(v)) \neq \pi(c'(z_i)) = i$, yielding that $c(v) \in L(v)$. Finally, for every $r$-tuple $(x_1, \ldots, x_r) \in \calF$, it holds that $(c'(x_1),\ldots,c'(x_r)) \in R$, and since $R$ is permutation-invariant, it follows that $(c(x_1),\ldots,c(x_r)) \in R$. This shows that $(G,L,\calF)$ is a $\YES$ instance of $R$-\CLC, as required.
\end{proof}

We are ready to derive Theorem~\ref{thm:Lower-R-CC}.

\begin{proof}[ of Theorem~\ref{thm:Lower-R-CC}]
For positive integers $q$, $r$, and $k \geq 3$, let $R \subseteq [q]^r$ be a permutation-invariant relation of arity $r$ over $[q]$ from which an \OR relation of arity $k$ is definable. By Lemma~\ref{lemma:non-listCC}, there exists a linear-parameter transformation from the $R$-\CLC problem to the $R$-\CC problem, both parameterized by the number of vertices $n$. If for some $\eps >0$, the latter admits a compression of size $O(n^{k-\eps})$, then by Lemma~\ref{lemma:composition}, so does the former. This implies, by Theorem~\ref{thm:Lower-R-CLC}, that $\NP \subseteq \coNPpoly$, and we are done.
\end{proof}

We conclude this section by highlighting the role of \OR-definability in the kernel complexity of $R$-\CC problems. Theorem~\ref{thm:Lower-R-CC} implies that if an \OR relation of arity $r$ is definable from a permutation-invariant relation $R$ of arity $r \geq 3$, then the trivial kernel for the $R$-\CC problem parameterized by the number of vertices is nearly optimal under a standard complexity assumption. As mentioned in the introduction, this definability condition is crucial, because if no \OR of arity $r$ is definable from $R$, a smaller kernel is achievable. This follows from a result of Carbonnel~\cite{Carbonnel22}, and we include here his argument for completeness.

\begin{proposition}[\cite{Carbonnel22}]\label{prop:Carbonnel}
For positive integers $q$ and $r \geq 3$, let $R \subseteq [q]^r$ be a relation of arity $r$ over $[q]$, such that no \OR relation of arity $r$ is definable from $R$. Then the $R$-\CC problem parameterized by the number of vertices $n$ admits a kernel with $O(n^{r-\eps})$ constraints for $\eps = 2^{1-r}$.
\end{proposition}

\begin{proof}
Let $R \subseteq [q]^r$ be a relation of arity $r \geq 3$ over $[q]$, such that no \OR relation of arity $r$ is definable from $R$. Consider an instance of the $R$-\CC problem, namely, a graph \mbox{$G=(V,E)$} on $n$ vertices and a collection $\calF \subseteq V^r$. We define an algorithm that, given such an input, initiates a collection $\calF'$ as $\calF$ and acts as follows. As long as there exist $r$ sets $A_1, \ldots, A_r \subseteq V$, of size $2$ each, such that all the $2^r$ tuples in the product $A_1 \times \cdots \times A_r$ lie in $\calF'$, it removes an arbitrary one of those tuples from $\calF'$. Finally, the algorithm returns the pair $(G,\calF')$.

It can be seen that the above algorithm can be implemented in polynomial time for a fixed integer $r$. Its correctness follows from the assumption that no \OR relation of arity $r$ is definable from $R$, which guarantees that any coloring of $G$ that satisfies $2^r-1$ of the $r$-tuples in a product $A_1 \times \cdots \times A_r$ of $2$-subsets $A_1, \ldots, A_r$ of $V$ must also satisfy the remaining one. Finally, to bound the number of constraints in $\calF'$, consider the $r$-uniform hypergraph on the vertex set $V \times [r]$, which includes, for each $r$-tuple $(x_1, \ldots, x_r) \in \calF'$, the edge $\{(x_i,i) \mid i \in [r]\}$. The algorithm ensures that this hypergraph contains no copy of the complete $r$-uniform $r$-partite hypergraph with vertex classes of size $2$, so a result of Erd\H{o}s~\cite{Erdos64} implies that the number of its edges is bounded by $O(n^{r-\eps})$ for $\eps = 2^{1-r}$. By $r \geq 3$, the total number of constraints in $(G,\calF')$, including the edges of $G$, is bounded by $O(n^{r-\eps})$, so we are done.
\end{proof}

\section{Lower Bounds for Uniformly Rainbow Free Coloring}\label{sec:URFClower}

In this section, we prove the lower bounds on the compression size of the $(d,\ell,q)$-\URFC problems stated in Theorem~\ref{thm:Full-URFC} (see Definition~\ref{def:URFC}).
We first observe that these problems fit the framework of $R$-\CC for appropriate relations $R$ and analyze the arity of the \OR relations definable from these relations. For $q \geq 3$, the lower bounds are deduced from Theorem~\ref{thm:Lower-R-CC}, and the case $q=2$ is treated separately.

\subsection{The Uniformly Rainbow Freeness Relation}

For positive integers $d$, $\ell$, and $q \geq d$, the $(d,\ell,q)$-\URFC problem can be realized as the $R$-\CC problem, where $R$ is the relation $\NUR_{d,\ell}^q$ (Uniformly Rainbow Freeness), defined as follows.

\begin{definition}\label{def:NUR}
For positive integers $d$, $\ell$, and $q \geq d$, we define the relations $\UR_{d,\ell}^q$ and $\NUR_{d,\ell}^q$ of arity $d \ell$ over $[q]$, whose $d \ell$-tuples are viewed as $d \times \ell$ matrices with entries indexed by $[d] \times [\ell]$.
Let $\UR_{d,\ell}^q \subseteq [q]^{d \times \ell}$ be the collection of all matrices $M \in [q]^{d \times \ell}$ in which each column consists of $d$ distinct elements and all columns contain the same set of elements, that is, the sets $\{M_{i,j} \mid i \in [d]\}$ for $j \in [\ell]$ are all equal and of size $d$. Further, let $\NUR_{d,\ell}^q = [q]^{d \times \ell} \setminus \UR_{d,\ell}^q$.
\end{definition}
\noindent
Note that the constraints representable by the relation $\NUR_{d,\ell}^q$ are precisely the uniformly rainbow freeness constraints allowed in the $(d,\ell,q)$-\URFC problem. This is because the order of vertices in each column does not affect membership in the relation, and because constraints with a repeated vertex in a column are always satisfied and thus can be omitted. Note further that the relation $\NUR_{d,\ell}^q$ is permutation-invariant.

The following lemma addresses the arity of \OR relations definable from $\NUR_{d,\ell}^q$.

\begin{lemma}\label{lemma:NURdefinable}
For all positive integers $d$ and $\ell$, the following statements hold.
\begin{enumerate}
  \item\label{itm:definable_q} If $\ell \geq 2$, then for all integers $q \geq d+2$, an \OR relation of arity $d \ell$ is definable from $\NUR_{d,\ell}^q$.
  \item\label{itm:definable_d+1} For all integers $q \geq d+1$, an \OR relation of arity $d \ell - 1$ is definable from $\NUR_{d,\ell}^{q}$.
  \item\label{itm:definable_d} For all integers $q \geq d$, an \OR relation of arity $(d-1) \ell$ is definable from $\NUR_{d,\ell}^q$.
\end{enumerate}
\end{lemma}

\begin{proof}
Fix positive integers $d$, $\ell$, and $q \geq d$. By Definition~\ref{def:OR definable}, to prove that an \OR relation of some arity $k$ is definable from the relation $\NUR_{d,\ell}^q$, we have to produce a set $J \subseteq [d] \times [\ell]$ of size $|J|=k$, a matrix $\alpha \in [q]^{d \times \ell}$, and elements $\beta_{i,j} \in [q]$ with $\beta_{i,j} \neq \alpha_{i,j}$ for all pairs $(i,j) \in J$, such that all matrices $M \in [q]^{d \times \ell}$ satisfying $M_{i,j} \in \{\alpha_{i,j},\beta_{i,j}\}$ for all $(i,j) \in J$ and $M_{i,j} = \alpha_{i,j}$ for all $(i,j) \notin J$ belong to $\NUR_{d,\ell}^q$, except the single matrix $\alpha$. For all three items of the lemma, we choose $\alpha \in [q]^{d \times \ell}$ as the matrix defined by $\alpha_{i,j} = i$ for all $(i,j) \in [d] \times [\ell]$. Note that each column of $\alpha$ consists of the $d$ elements of $[d]$, hence $\alpha \notin \NUR_{d,\ell}^q$. We proceed by handling each item of the lemma separately (see Figure~\ref{fig:alpha_beta} for an illustration).

For Item~\ref{itm:definable_q}, suppose that $\ell \geq 2$ and $q \geq d+2$, set $J = [d] \times [\ell]$, and for each $(i,j) \in J$, define
\[
\beta_{i,j} =
    \begin{cases}
        d+1, & \text{if } j = 1,\\
        d+2, & \text{if } j \ge 2.
    \end{cases}
\]
Note that $\alpha_{i,j} \neq \beta_{i,j}$ for all $(i,j) \in J$.
Let $M \in [q]^{d \times \ell}$ be a matrix such that $M_{i,j} \in \{\alpha_{i,j},\beta_{i,j}\}$ for all $(i,j) \in J$ and $M \neq \alpha$.
It follows that there exists a pair $(i,j) \in [d] \times [\ell]$ for which \mbox{$M_{i,j} = \beta_{i,j} \in \{d+1,d+2\}$}. However, the value $d+1$ is missing from the second column of $M$, and the value $d+2$ is missing from the first column of $M$. Therefore, the sets of values in the columns of $M$ are not all equal, implying that $M \in \NUR_{d,\ell}^q$. Since $|J| = d\ell$, it follows that an \OR relation of arity $d\ell$ is definable from $\NUR_{d,\ell}^q$, as required.

For Item~\ref{itm:definable_d+1}, suppose that $q \geq d+1$, and let $J = ([d] \times [\ell]) \setminus \{(1,1)\}$ be the set of all entries except the entry in the first row and first column. For each $(i,j) \in J$, define
\[
\beta_{i,j} =
    \begin{cases}
        d+1, & \text{if } i = 1,\\
        1, & \text{otherwise.}
    \end{cases}
\]
We clearly have $\alpha_{i,j} \neq \beta_{i,j}$ for all $(i,j) \in J$.
Let $M \in [q]^{d \times \ell}$ be a matrix with \mbox{$M_{i,j} \in \{\alpha_{i,j},\beta_{i,j}\}$} for all $(i,j) \in J$, $M_{1,1} = \alpha_{1,1}=1$, and $M \neq \alpha$. It follows that there exists a pair $(i,j) \in J$ such that $M_{i,j}=\beta_{i,j}$. If for some $j \in [\ell] \setminus \{1\}$, it holds that $M_{1,j} = \beta_{1,j} = d+1$, then $d+1$ appears in the $j$th column of $M$ but is missing from its first column. Otherwise, we have $M_{1,j}=\alpha_{1,j}=1$ for all $j \in [\ell]$, and there exist $i \in [d] \setminus \{1\}$ and $j \in [\ell]$ such that $M_{i,j} = \beta_{i,j} = 1$, so the value $1$ appears at least twice in the $j$th column of $M$. In both cases, we have $M \in \NUR_{d,\ell}^{q}$. Therefore, by $|J| = d\ell-1$, an \OR relation of arity $d\ell-1$ is definable from $\NUR_{d,\ell}^{q}$, as required.

Finally, for Item~\ref{itm:definable_d}, suppose that $q \geq d$, let $J = ([d] \setminus \{1\}) \times [\ell]$ be the set of all entries except those of the first row, and define $\beta_{i,j} = 1$ for all pairs $(i,j) \in J$. It holds that $\alpha_{i,j} \neq \beta_{i,j}$ for all $(i,j) \in J$. Let $M \in [q]^{d \times \ell}$ be a matrix with $M_{i,j} \in \{\alpha_{i,j},\beta_{i,j}\}$ for all $(i,j) \in J$, $M_{i,j} = \alpha_{i,j}$ for all $(i,j) \notin J$, and $M \neq \alpha$. It follows that there exists a pair $(i,j) \in J$ such that $M_{i,j}=\beta_{i,j}=1$. However, we have $M_{1,j} = \alpha_{1,j} = 1$, hence the value $1$ appears at least twice in the $j$th column of $M$, implying that $M \in \NUR_{d,\ell}^q$. By $|J| = (d-1)\ell$, we obtain that an \OR relation of arity $(d-1) \ell$ is definable from $\NUR_{d,\ell}^q$. This completes the proof.
\end{proof}

\begin{figure}[htb]
\[
\begin{array}{c@{\quad}c@{\quad}c@{\quad}c}
\alpha &
\beta\ \text{for } q \ge 5 &
\beta\ \text{for } q = 4 &
\beta\ \text{for } q = 3 \\[2mm]
\begin{pmatrix}
1 & 1 & 1 & 1 \\
2 & 2 & 2 & 2 \\
3 & 3 & 3 & 3
\end{pmatrix} &
\begin{pmatrix}
4 & 5 & 5 & 5 \\
4 & 5 & 5 & 5 \\
4 & 5 & 5 & 5
\end{pmatrix} &
\begin{pmatrix}
- & 4 & 4 & 4 \\
1 & 1 & 1 & 1 \\
1 & 1 & 1 & 1
\end{pmatrix} &
\begin{pmatrix}
- & - & - & - \\
1 & 1 & 1 & 1 \\
1 & 1 & 1 & 1
\end{pmatrix}
\end{array}
\]
\caption{An illustration of the proof of Lemma~\ref{lemma:NURdefinable} for $d=3$ and $\ell=4$. The matrix $\alpha$ and the partial matrices $\beta$ are depicted for all values of $q \geq 3$, with a dash indicating an entry outside $J$. The arity of the \OR relation definable from $\NUR_{d,\ell}^q$ is equal to the number of non-dash entries in $\beta$.}
\label{fig:alpha_beta}
\end{figure}

\subsection{\texorpdfstring{Uniformly Rainbow Free $2$-Coloring}{Uniformly Rainbow Free 2-Coloring}}

We now focus on the $(d,\ell,q)$-\URFC problems with $q=2$. We obtain the following lower bound on their compression size.

\begin{proposition}\label{prop:q=2}
For every integer $k \geq 3$, each of the $(1,k,2)$- and $(2,k-1,2)$-\URFC problems parameterized by the number of vertices $n$ does not admit a compression of size $O(n^{k-1-\eps})$ unless $\NP \subseteq \coNPpoly$.
\end{proposition}

Proposition~\ref{prop:q=2} stems from the lower bounds on the compression size of the $k$-\NAE problems, stated in Theorem~\ref{thm:NAE}, via the following lemma.

\begin{lemma}\label{lemma:NAE<URFC_2}
For every integer $k \geq 2$, there exist linear-parameter transformations from the $k$-\NAE problem, parameterized by the number of variables, to each of the following problems, parameterized by the number of vertices.
\begin{enumerate}
  \item\label{itm:1,k,2} $(1,k,2)$-\URFC.
  \item\label{itm:2,k-1,2} $(2,k-1,2)$-\URFC.
\end{enumerate}
\end{lemma}

\begin{proof}
Fix an integer $k \geq 2$. Let $\varphi$ be an instance of the $k$-\NAE problem, that is, a $k$-CNF formula on $n$ variables, denoted by $x_1, \ldots, x_n$. We define a transformation that, given such an instance, produces an instance $(G,\calF)$ as follows. The graph $G$ is a perfect matching on $2n$ vertices, identified with the literals in
\[\{x_i \mid i \in [n]\} \cup \{\overline{x_i} \mid i \in [n]\},\]
where each variable vertex $x_i$ is adjacent to the vertex $\overline{x_i}$ corresponding to its negation. For each clause $(y_{1} \vee \cdots \vee y_{k})$ in $\varphi$, we add to $\calF$ a tuple, whose definition depends on the considered item of the lemma. For Item~\ref{itm:1,k,2}, we add the $k$-tuple
\[(\{y_{1}\}, \{y_{2}\},\ldots, \{y_{k}\}).\]
For Item~\ref{itm:2,k-1,2}, we add the $(k-1)$-tuple
\[(\{y_{1},\overline{y_{2}}\}, \{y_{1},\overline{y_{3}}\}, \ldots, \{y_{1},\overline{y_{k}}\}),\]
where negation swaps a variable vertex with the vertex representing its negation, and vice versa.
This completes the description of the instance $(G,\calF)$, which can be constructed in polynomial time and forms a valid instance of the corresponding \URFC problem. The transformation is linear-parameter, because $G$ has $2n$ vertices.

The correctness of the above transformations relies on viewing every truth assignment for the variables of $\varphi$ as a proper coloring of $G$ with colors $1$ and $2$, assigning color $1$ to literals set to $\true$ and color $2$ to literals set to $\false$. It is straightforward to verify that a truth assignment satisfies a clause in $\varphi$ in the NAE sense if and only if the associated tuple in $\calF$ is not uniformly rainbow with respect to the corresponding coloring of $G$. For Item~\ref{itm:1,k,2}, this follows directly from the definition of NAE-satisfiability. For Item~\ref{itm:2,k-1,2}, notice that a truth assignment satisfies a clause $(y_{1} \vee \cdots \vee y_{k})$ in the NAE sense if and only if $y_{1}$ shares its value with $\overline{y_{j}}$ for some $j \in [k] \setminus \{1\}$.
\end{proof}

Proposition~\ref{prop:q=2} follows by combining Lemma~\ref{lemma:composition}, Theorem~\ref{thm:NAE}, and Lemma~\ref{lemma:NAE<URFC_2}.
We note that it can also be derived from the results of~\cite{ChenJP20} on Boolean-domain constraint satisfaction problems.

\subsection{Proof of Lower Bounds}

We prove now the lower bounds asserted in Theorem~\ref{thm:Full-URFC}, restated below.

\begin{theorem}\label{thm:Lower-URFC}
For positive integers $d$, $\ell$, and $q \geq d$, set $\eta = \eta(d,\ell,q)$ as in Definition~\ref{def:eta}.
For any real $\eps>0$, the $(d,\ell,q)$-\URFC problem parameterized by the number of vertices $n$ does not admit a compression of size $O(n^{\eta-\eps})$ unless \mbox{$\NP \subseteq \coNPpoly$}.
\end{theorem}

In the proof, we will use the following immediate consequence of a result from~\cite{JansenP19color} to obtain a near-quadratic lower bound for some exceptional cases.

\begin{proposition}\label{prop:URFC_quad}
For positive integers $d$, $\ell$, and $q \geq 3$, the $(d,\ell,q)$-\URFC problem parameterized by the number of vertices $n$ does not admit a compression of size $O(n^{2-\eps})$ unless $\NP \subseteq \coNPpoly$.
\end{proposition}

\begin{proof}
Fix positive integers $d$, $\ell$, and $q \geq 3$. The $q$-\Col problem admits a linear-parameter transformation into the $(d,\ell,q)$-\URFC problem, when both are parameterized by the number of vertices, which simply maps a given graph $G$ to the pair $(G,\emptyset)$. A result of~\cite{JansenP19color} implies that for every integer $q \geq 3$, the $q$-\Col problem parameterized by the number of vertices $n$ does not admit a compression of size $O(n^{2-\eps})$ unless $\NP \subseteq \coNPpoly$. By Lemma~\ref{lemma:composition}, this completes the proof.
\end{proof}

We proceed to prove Theorem~\ref{thm:Lower-URFC}.

\begin{proof}[ of Theorem~\ref{thm:Lower-URFC}]
Fix positive integers $d$, $\ell$, and $q \geq d$, and set $\eta = \eta(d,\ell,q)$. Recall that the $(d,\ell,q)$-\URFC problem can be expressed as the $\NUR_{d,\ell}^q$-\CC problem and that the relation $\NUR_{d,\ell}^q$ is permutation-invariant. To prove the theorem, we examine each case in the definition of $\eta(d,\ell,q)$.
\begin{itemize}
  \item Suppose that $\ell \geq 2$ and $q \geq d+2$, which implies that $\eta = d\ell \geq 2$. By Item~\ref{itm:definable_q} of Lemma~\ref{lemma:NURdefinable}, an \OR relation of arity $\eta$ is definable from $\NUR_{d,\ell}^q$. Therefore, if $\eta \geq 3$, the result follows from Theorem~\ref{thm:Lower-R-CC}, and if $\eta = 2$, since $q \geq 3$, the result follows from Proposition~\ref{prop:URFC_quad}.
  \item Suppose that $\ell \geq 2$, $q = d+1$, and $(d,\ell) \neq (1,2)$, which implies that $\eta = d\ell-1 \geq 2$. By Item~\ref{itm:definable_d+1} of Lemma~\ref{lemma:NURdefinable}, an \OR relation of arity $\eta$ is definable from $\NUR_{d,\ell}^q$. Therefore, if $\eta \geq 3$, the result follows from Theorem~\ref{thm:Lower-R-CC}, and if $\eta = 2$, then $(d,\ell,q)=(1,3,2)$, so the result follows from Proposition~\ref{prop:q=2}.
  \item Suppose that either $\ell \geq 2$ and $q=d \geq 2$, or $\ell=1$ and $d \geq 3$, which implies that \mbox{$\eta = (d-1)\ell \geq 2$}. By Item~\ref{itm:definable_d} of Lemma~\ref{lemma:NURdefinable}, an \OR relation of arity $\eta$ is definable from $\NUR_{d,\ell}^q$. Therefore, if $\eta \geq 3$, the result follows from Theorem~\ref{thm:Lower-R-CC}, and if $\eta = 2$, the triple $(d,\ell,q)$ is either $(2,2,2)$ or $(3,1,q)$ for some $q \geq 3$. In the former case, the result follows from Proposition~\ref{prop:q=2}, and in the latter, from Proposition~\ref{prop:URFC_quad}.
  \item Suppose that $\ell =1$, $d \leq 2$, and $q \geq 3$, which implies that $\eta = 2$. The result again follows from Proposition~\ref{prop:URFC_quad}.
  \item In the final case, when either $q=2$ and $d\ell \leq 2$, or $q=1$, we have $\eta = 0$, so there is nothing to prove.
\end{itemize}
The proof is complete.
\end{proof}

\section{Upper Bounds for Uniformly Rainbow Free Coloring}\label{sec:URFCupper}

In this section, we prove the upper bounds on the kernel size of the $(d,\ell,q)$-\URFC problems, stated in Theorem~\ref{thm:Full-URFC}. For the non-trivial cases, where the degree of the polynomial bounding the kernel size is smaller than $d\ell$, we employ the approach of~\cite{JansenP17}, which relies on constructions of suitable low-degree polynomials. We first describe the method and the specific type of polynomials required for our purposes, and then construct the desired polynomials and derive the resulting bounds. We also extend the bounds to a generalized version of the $(d,\ell,q)$-\URFC problems that permits combining constraints with different values of $d$ and $\ell$.

\subsection{Uniformly Rainbow Property via Polynomials}

Our upper bounds rely on constructions of polynomials that capture the uniformly rainbow property, as defined below.
\begin{definition}\label{def:URproperty}
For positive integers $m$ and $q$, and for a field $\Fset$, let $C \subseteq \Fset^m$ be a set of size $|C|=q$. For a positive integer $n$, a matrix $M = (a_1, \ldots, a_n) \in \Fset^{m \times n}$ is called {\em $C$-colored} if each of its columns belongs to $C$, that is, $a_i \in C$ for all $i \in [n]$. For positive integers $d$ and $\ell$, let $p: \Fset^{m \times (d\ell)} \to \Fset$ be a polynomial in $m d \ell$ variables, which are represented as an $m \times (d \ell)$ matrix over $\Fset$. We say that the pair $(C,p)$ {\em captures the $(d,\ell,q)$-uniformly rainbow property} if for every $C$-colored matrix $M = (a_1, \ldots, a_{d \ell}) \in \Fset^{m \times (d \ell)}$, it holds that $p(M) \neq 0$ if and only if the $\ell$ disjoint blocks of $d$ consecutive columns in $M$ coincide as sets and have size $d$, equivalently, the $\ell$ sets
\[F_j = \{a_i \mid (j-1)d < i \leq jd\},~~~j \in [\ell],\]
are all equal and of size $d$.
\end{definition}

The following theorem shows that a set of vectors and a low-degree polynomial capturing the $(d,\ell,q)$-uniformly rainbow property yield an economical kernel for the $(d,\ell,q)$-\URFC problem. In what follows, the {\em solution set} of an instance $(G,\calF)$ of the problem refers to the set of all proper colorings of $G$ with $q$ colors for which no tuple in $\calF$ is uniformly rainbow. Further, a field $\Fset$ is said to be {\em efficient} if field operations and Gaussian elimination over $\Fset$ can be performed in polynomial time in the size of a reasonable input encoding. For example, all finite fields and the real field $\R$ restricted to rationals are efficient.

\begin{theorem}\label{thm:JP-Upper}
Let $d, \ell, m, q \geq 1$ and $r \geq 2$ be integers, and let $\Fset$ be an efficient field.
Suppose that there exist a set $C \subseteq \Fset^m$ with $|C|=q$ and a polynomial $p: \Fset^{m \times (d \ell)} \to \Fset$ of degree at most $r$, such that the pair $(C,p)$ captures the $(d,\ell,q)$-uniformly rainbow property. Then the $(d,\ell,q)$-\URFC problem parameterized by the number of vertices $n$ admits a kernel with $O(n^r)$ constraints.
Moreover, given an instance $(G,\calF)$, the kernel returns an instance $(G,\calF')$ with $\calF' \subseteq \calF$, such that $(G,\calF)$ and $(G,\calF')$ have the same solution set.
\end{theorem}

\begin{proof}
Let $d,\ell,m,q \geq 1$ and $r \geq 2$ be integers, let $\Fset$ be an efficient field, and consider a set \mbox{$C = \{a_1, \ldots, a_q\} \subseteq \Fset^m$} and a polynomial $p: \Fset^{m \times (d \ell)} \to \Fset$ of degree at most $r$, such that $(C,p)$ captures the $(d,\ell,q)$-uniformly rainbow property. An instance of the $(d,\ell,q)$-\URFC problem consists of a graph $G=(V,E)$ on $n$ vertices and a collection $\calF$ of $\ell$-tuples of $d$-subsets of $V$. We define a kernelization algorithm that, given such an input, returns the same graph $G$ along with a collection $\calF' \subseteq \calF$ produced by the following steps.
\begin{enumerate}
  \item Associate with each vertex $v \in V$ an $m$-dimensional symbolic variable vector $x_v$ over the field $\Fset$. Note that the total number of variables is $m \cdot n$.
  \item For each $\ell$-tuple $F = (F_1, \ldots, F_\ell) \in \calF$, let $p_F: \Fset^{m \times (d \ell)} \to \Fset$ denote the polynomial obtained from $p$ by substituting a matrix whose columns are the variable vectors associated with the vertices of the sets $F_1, \ldots, F_\ell$, in this order, where the vectors corresponding to each set are ordered arbitrarily.
  \item Compute a collection $\calF' \subseteq \calF$, such that the set $\{p_F \mid F \in \calF'\}$ forms a basis for the subspace of polynomials spanned by $\{p_F \mid F \in \calF\}$.
\end{enumerate}

We note that the polynomials $p_F$ with $F \in \calF$ all lie in the space of polynomials of degree at most $r$ in $m \cdot n$ variables over $\Fset$, which has dimension $\binom{m \cdot n+r}{r} = O(n^r)$. Since the set $\{p_F \mid F \in \calF'\}$ forms a basis of a subspace of this space, we conclude that $|\calF'| \leq O(n^r)$. By $r \geq 2$, it follows that the total number of constraints in the produced instance, including the edges in $G$, is bounded by $O(n^r)$. Moreover, the collection $\calF'$ can be computed in polynomial time by performing Gaussian elimination over the efficient field $\Fset$ on a system with $O(n^r)$ variables.

We finally show that $(G,\calF)$ and $(G,\calF')$ have the same solution set, which, in particular, yields the correctness of the kernelization algorithm.
Clearly, by $\calF' \subseteq \calF$, any solution for $(G,\calF)$ is also a solution for $(G,\calF')$.
Conversely, consider a solution for $(G,\calF')$, i.e., a proper coloring $c:V \to [q]$ of $G$, such that no $\ell$-tuple of $\calF'$ is uniformly rainbow.
We define an assignment of elements from $\Fset$ to the variables of $(x_v)_{v \in V}$ by setting each vector $x_v$ with $v \in V$ to the vector $a_{c(v)} \in C$ corresponding to its color. Since $(C,p)$ captures the $(d,\ell,q)$-uniformly rainbow property and no $\ell$-tuple of $\calF'$ is uniformly rainbow, all polynomials $p_F$ with $F \in \calF'$ evaluate to zero on this assignment, and since they form a basis of the subspace spanned by $\{p_F \mid F \in \calF\}$, we conclude that all polynomials $p_F$ with $F \in \calF$ vanish on this assignment as well. Since $(C,p)$ captures the $(d,\ell,q)$-uniformly rainbow property, this implies that no $\ell$-tuple in $\calF$ is uniformly rainbow with respect to $c$, hence $c$ is a solution for $(G,\calF)$, as required.
\end{proof}

The following lemma supplies the vector sets and low-degree polynomials that capture the $(d,\ell,q)$-uniformly rainbow property and yield the non-trivial upper bounds in Theorem~\ref{thm:Full-URFC}. It adapts constructions from the prior works~\cite{JansenP19color,HavivR24,HavivR25}.

\begin{lemma}\label{lemma:C,p}
For all positive integers $d$, $\ell$, $m$, $q$, and $r$ satisfying one of the following conditions, and for every field $\Fset$ with $|\Fset| \geq q$, there exist a set $C \subseteq \Fset^m$ with $|C|=q$ and a polynomial $p: \Fset^{m \times (d \ell)} \to \Fset$ of degree at most $r$, such that the pair $(C,p)$ captures the $(d,\ell,q)$-uniformly rainbow property.
\begin{enumerate}
  \item\label{itm:pol:upper_l=1} $\ell =1$, $q \geq d = m$, and $r = d-1$.
  \item\label{itm:pol:upper_q=d} $m = q = d$ and $r = (d-1) \ell$.
  \item\label{itm:pol:upper_q=d+1} $m = q = d+1$ and $r = d \ell-1$.
\end{enumerate}
\end{lemma}

\begin{proof}
Let $d$, $\ell$, $m$, and $q$ be positive integers satisfying the assumption of the lemma, let $\Fset$ be a field with $|\Fset| \geq q$, and let $\alpha_1, \ldots, \alpha_q$ be distinct field elements. We start with some notations. Let $C_{m,q} \subseteq \Fset^m$ denote the set of vectors $(1,\alpha_i,\alpha_i^2,\ldots, \alpha_i^{m-1})^t$ for $i \in [q]$, and note that $|C_{m,q}|=q$. For every positive integer $t \leq m$, the vectors of any $t$-subset of $C_{m,q}$ restricted to their first $t$ entries are distinct columns of a $t \times t$ Vandermonde matrix and are therefore linearly independent over $\Fset$. For such an integer $t$, let $p_{m,t}:\Fset^{m \times t} \to \Fset$ denote the polynomial that maps a matrix $M \in \Fset^{m \times t}$ to the determinant of the $t \times t$ submatrix induced by the first $t$ rows of $M$, after replacing the entries of the first row by ones. Since the determinant polynomial on $t \times t$ matrices is a linear combination of products of $t$ variables, one from each row, it follows that the degree of $p_{m,t}$ is $t-1$. For every $C_{m,q}$-colored matrix $M \in \Fset^{m \times t}$, all entries in the first row are ones, so it holds that $p_{m,t}(M) \neq 0$ if and only if the $t$ columns of $M$ are pairwise distinct. We now handle each item of the lemma separately.

For Item~\ref{itm:pol:upper_l=1}, suppose that $\ell =1$ and $q \geq d = m$, and consider the set $C_{d,q} \subseteq \Fset^{d}$ and the polynomial $p_{d,d}:\Fset^{d \times d} \to \Fset$, whose degree is $d-1$. For every $C_{d,q}$-colored matrix $M \in \Fset^{d \times d}$, we have $p_{d,d}(M) \neq 0$ if and only if the columns of $M$ are pairwise distinct, hence the pair $(C_{d,q},p_{d,d})$ captures the $(d,1,q)$-uniformly rainbow property, as required.

For Item~\ref{itm:pol:upper_q=d}, suppose that $m = q = d$, and consider the set $C_{d,d} \subseteq \Fset^{d}$ and the polynomial \mbox{$p_{d,d}:\Fset^{d \times d} \to \Fset$}, whose degree is $d-1$. Let $p:\Fset^{d \times (d\ell)} \to \Fset$ denote the polynomial that maps a given matrix \mbox{$M=(M_1, \ldots, M_\ell) \in \Fset^{d \times (d\ell)}$}, where $M_i \in \Fset^{d \times d}$ for each $i \in [\ell]$, to the product $\prod_{i \in [\ell]}{p_{d,d}(M_i)}$. The degree of $p$ is clearly $(d-1) \ell$. We claim that the pair $(C_{d,d},p)$ captures the $(d,\ell,q)$-uniformly rainbow property. To see this, let $M=(M_1, \ldots, M_\ell) \in \Fset^{d \times (d\ell)}$ be a $C_{d,d}$-colored matrix. It follows that $p(M) \neq 0$ precisely when $p_{d,d}(M_i) \neq 0$ for all $i \in [\ell]$, which occurs if and only if the $d$ columns of each $M_i$ are all the $d$ vectors of $C_{d,d}$, as required.

For Item~\ref{itm:pol:upper_q=d+1}, suppose that $m = q = d+1$, and consider the set $C_{q,q} \subseteq \Fset^{q}$ and the polynomials $p_{q,d}:\Fset^{q \times d} \to \Fset$ and $p_{q,q}:\Fset^{q \times q} \to \Fset$. Further, let $a \in \Fset^q$ denote the sum of the vectors in $C_{q,q}$.
We define a polynomial $p:\Fset^{q \times (d\ell)} \to \Fset$ as follows. Consider a matrix $M=(M_1, \ldots, M_\ell) \in \Fset^{q \times (d \ell)}$, where $M_i \in \Fset^{q \times d}$ for each $i \in [\ell]$, and let $y$ denote the sum of the columns of $M_1$. The polynomial $p$ maps such a matrix $M$ to \[p_{q,d}(M_1) \cdot \prod_{i=2}^{\ell}{p_{q,q}(M_i,a-y)}.\] Substitutions of linear terms in a polynomial do not increase its degree, hence the degree of $p$ is at most $(d-1) + (q-1)(\ell-1) = (d-1) + d(\ell-1) = d \ell -1$.

We now show that the pair $(C_{q,q},p)$ captures the $(d,\ell,q)$-uniformly rainbow property.
For a $C_{q,q}$-colored matrix \mbox{$M=(M_1, \ldots, M_\ell) \in \Fset^{q \times (d\ell)}$}, it holds that $p(M) \neq 0$ if and only if $p_{q,d}(M_1) \neq 0$ and $p_{q,q}(M_i,a-y) \neq 0$ for all $i \in [\ell] \setminus \{1\}$. This occurs precisely when the columns of $M_1$ are distinct and the vector $a-y$, the unique vector of $C_{q,q}$ missing from $M_1$, completes the column set of each $M_i$ with $i \in [\ell] \setminus \{1\}$ to the entire set $C_{q,q}$. Therefore, $p(M) \neq 0$ if and only if the column sets of the matrices $M_1, \ldots, M_\ell$ are equal and of size $d$, so we are done.
\end{proof}

\subsection{Proof of Upper Bounds}

The following theorem confirms the upper bounds asserted in Theorem~\ref{thm:Full-URFC}.
\begin{theorem}\label{thm:Upper-URFC}
For positive integers $d$, $\ell$, and $q \geq d$, set $\eta = \eta(d,\ell,q)$ as in Definition~\ref{def:eta}.
The $(d,\ell,q)$-\URFC problem parameterized by the number of vertices $n$ admits a kernel with $O(n^{\eta})$ constraints.
Moreover, when $\eta \geq 2$, given an instance $(G,\calF)$, the kernel returns an instance $(G,\calF')$ with $\calF' \subseteq \calF$, such that $(G,\calF)$ and $(G,\calF')$ have the same solution set.
\end{theorem}

\begin{proof}
Fix positive integers $d$, $\ell$, and $q \geq d$, and set $\eta = \eta(d,\ell,q)$. To prove the theorem, we examine each case in the definition of $\eta(d,\ell,q)$.
\begin{itemize}
  \item Suppose that $\ell \geq 2$ and $q \geq d+2$, which implies that $\eta = d\ell \geq 2$. In this case, the trivial kernel that removes repeated constraints gives an instance with $O(n^{\eta})$ constraints.
  \item Suppose that $\ell \geq 2$, $q = d+1$, and $(d,\ell) \neq (1,2)$, which implies that $\eta = d\ell-1 \geq 2$. Let $\Fset$ be a fixed finite field with $|\Fset| \geq q$. By Item~\ref{itm:pol:upper_q=d+1} of Lemma~\ref{lemma:C,p}, there exist a set $C \subseteq \Fset^q$ with $|C|=q$ and a polynomial $p: \Fset^{q \times (d \ell)} \to \Fset$ of degree at most $\eta$, such that the pair $(C,p)$ captures the $(d,\ell,q)$-uniformly rainbow property. Since $\Fset$ is finite and thus efficient, the result follows from Theorem~\ref{thm:JP-Upper}.
  \item Suppose that either $\ell \geq 2$ and $q=d \geq 2$, or $\ell=1$ and $d \geq 3$, which implies that \mbox{$\eta = (d-1)\ell \geq 2$}. Let $\Fset$ be a fixed finite field with $|\Fset| \geq q$. By Items~\ref{itm:pol:upper_l=1} and~\ref{itm:pol:upper_q=d} of Lemma~\ref{lemma:C,p}, there exist a set $C \subseteq \Fset^d$ with $|C|=q$ and a polynomial $p: \Fset^{d \times (d \ell)} \to \Fset$ of degree at most $\eta$, such that the pair $(C,p)$ captures the $(d,\ell,q)$-uniformly rainbow property. As before, since $\Fset$ is finite and thus efficient, the result follows from Theorem~\ref{thm:JP-Upper}.
  \item Suppose that $\ell =1$, $d \leq 2$, and $q \geq 3$, which implies that $\eta = 2$. Here, the trivial kernel that removes repeated constraints gives an instance with $O(n^{\eta})$ constraints.
  \item In the final case, we have either $q=2$ and $d\ell \leq 2$, or $q=1$, and thus $\eta = 0$. When $q=1$, the problem is equivalent to testing whether the input graph is edgeless and, additionally, when $d=1$, whether the collection of constraints is empty. When $q=2$ and $d\ell \leq 2$, the problem reduces to the standard $2$-\Col problem. Indeed, for $(d,\ell)=(1,2)$, a constraint $(\{x\},\{y\})$ can be represented by an edge connecting $x$ and $y$. For $(d,\ell)=(2,1)$, a constraint $(\{x,y\})$ can be handled by merging the vertices $x$ and $y$ into a single vertex, and for $(d,\ell)=(1,1)$, no uniformly rainbow freeness constraint can be satisfied. In all these cases, the problem is solvable in polynomial time and thus admits a kernel of constant size.
\end{itemize}

We finally note that the condition $\eta \geq 2$ holds in the first four cases above. In all of them, the kernel is obtained either by Theorem~\ref{thm:JP-Upper} or by removing duplicate constraints. Therefore, given an instance $(G,\calF)$, the kernel returns an instance $(G,\calF')$ with $\calF' \subseteq \calF$, such that $(G,\calF)$ and $(G,\calF')$ have the same solution set, as required.
\end{proof}

\begin{remark}\label{remark:encoding}
Theorem~\ref{thm:Upper-URFC} provides an upper bound of $O(n^{\eta})$ on the number of constraints in the kernelized instances of the $(d,\ell,q)$-\URFC problem parameterized by the number of vertices $n$. This also implies an upper bound on the encoding size of these instances. To see this, let us focus on the cases with $\eta \geq 2$, and consider a kernelized instance $(G,\calF')$. The adjacencies in $G$ can be encoded in $O(n^2)$ bits. Further, when $\eta = d \ell$, the constraints in $\calF'$ can be encoded in $O(n^{\eta})$ bits, one per each potential constraint, so the total encoding size is $O(n^\eta)$. When $\eta < d \ell$, each constraint in $\calF'$ can be encoded in $O(\log n)$ bits, representing the vertices of the corresponding tuple, so the total encoding size is $O(n^\eta \cdot \log n)$.
\end{remark}

\subsection{Generalized Uniformly Rainbow Free Coloring}

We finally turn to an extension of the $(d,\ell,q)$-\URFC problems, in which the constraints can involve various values of $d$ and $\ell$.

\begin{definition}\label{def:GURFC}
For positive integers $m$ and $q$, let $I = \{ (d_i,\ell_i) \mid i \in [m]\}$ be a set of $m$ pairs of positive integers with $q \geq \max_{i \in [m]}{d_i}$.
The input of the $(I,q)$-\GURFC problem consists of a graph $G=(V,E)$ and a collection $\calF = \cup_{i \in [m]}{\calF_i}$, where for each $i \in [m]$, $\calF_i$ is a collection of $\ell_i$-tuples of $d_i$-subsets of $V$.
The goal is to decide whether there exists a proper coloring $c:V \to [q]$ of $G$, such that no tuple in $\calF$ is uniformly rainbow with respect to $c$.
\end{definition}

We prove the following result.

\begin{theorem}\label{thm:GURFC}
For positive integers $m$ and $q$, let $I = \{ (d_i,\ell_i) \mid i \in [m]\}$ be a set of $m$ pairs of positive integers with $q \geq \max_{i \in [m]}{d_i}$. For each $i \in [m]$, let $\eta_i = \eta(d_i,\ell_i,q)$ be as in Definition~\ref{def:eta}, and suppose that $\eta_i \geq 2$ for all $i \in [m]$. Set $\eta = \max_{i \in [m]} {\eta_i}$. Then the $(I,q)$-\GURFC problem parameterized by the number of vertices $n$ admits a kernel with $O(n^{\eta})$ constraints. Moreover, given an instance $(G,\calF)$, the kernel returns an instance $(G,\calF')$ with $\calF' \subseteq \calF$, such that $(G,\calF)$ and $(G,\calF')$ have the same solution set.
\end{theorem}

\begin{proof}
For positive integers $m$ and $q$, let $I = \{ (d_i,\ell_i) \mid i \in [m]\}$, and let $\eta_i$ with $i \in [m]$ and $\eta$ be as in the statement of the theorem. The input of the $(I,q)$-\GURFC problem consists of a graph $G=(V,E)$ on $n$ vertices and a collection \mbox{$\calF = \cup_{i \in [m]}{\calF_i}$}, where for each $i \in [m]$, $\calF_i$ is a collection of $\ell_i$-tuples of $d_i$-subsets of $V$.
Consider the algorithm that, given such an input, for each $i \in [m]$, applies the algorithm from Theorem~\ref{thm:Upper-URFC} for the $(d_i,\ell_i,q)$-\URFC problem to the instance $(G,\calF_i)$. By $\eta_i \geq 2$, this gives a collection $\calF'_i \subseteq \calF_i$ with $|\calF'_i| \leq O(n^{\eta_i})$, such that $(G,\calF_i)$ and $(G,\calF'_i)$ have identical solution sets. Our algorithm then returns the graph $G$ along with the collection $\calF' = \cup_{i \in [m]}{\calF'_i} \subseteq \calF$. Since each algorithm from Theorem~\ref{thm:Upper-URFC} runs in polynomial time, the overall algorithm does as well.
The total number of constraints in the produced instance $(G,\calF')$ is
\[|E|+|\calF'| = |E|+\sum_{i \in [m]}{|\calF'_i|} \leq |E|+\sum_{i \in [m]}{O(n^{\eta_i})} \leq O(n^2)+m \cdot O(n^{\eta}) = O(n^\eta).\]

For correctness, notice that the solution set of $(G,\calF)$ is the intersection of all solution sets of the $(d_i,\ell_i,q)$-\URFC instances $(G,\calF_i)$ with $i \in [m]$, which by the guarantee of Theorem~\ref{thm:Upper-URFC}, is equal to the intersection of all solution sets of the $(d_i,\ell_i,q)$-\URFC instances $(G,\calF'_i)$ with $i \in [m]$. The latter is precisely the solution set of $(G,\calF')$. It follows that $(G,\calF)$ and $(G,\calF')$ have the same solution set, implying the correctness of the kernel.
\end{proof}

\section{Applications}\label{sec:applications}

In this section, we present two applications of our results on the uniformly rainbow free coloring problems. The first application addresses uniform hypergraph coloring problems parameterized by the number of vertices, and the second concerns graph coloring problems parameterized by the vertex-deletion distance to a disjoint union of cliques.

\subsection{Uniform Hypergraph Coloring}\label{sec:hyper}

For positive integers $\ell$ and $q$, the $q$-\Col problem on $\ell$-uniform hypergraphs asks whether a given $\ell$-uniform hypergraph admits a proper coloring with $q$ colors. We prove here Theorem~\ref{thm:Hyper}, which asserts that for all integers $\ell \geq 2$ and $q \geq 3$ and for any real $\eps >0$, the problem does not admit a compression of size $O(n^{\ell-\eps})$ when parameterized by the number of vertices $n$, unless \mbox{$\NP \subseteq \coNPpoly$}. This shows that the trivial kernel of this problem, whose size is $O(n^\ell)$, is likely to be near-optimal. Previously, the result was known to hold for $\ell \in \{2,3\}$, as shown in~\cite{JansenP19color,Beukers21}.

We start with the following lemma.

\begin{lemma}\label{lemma:URFC<qCol}
For all integers $\ell \geq 2$ and $q \geq 2$, there exists a linear-parameter transformation from the $(1,\ell,q)$-\URFC problem to the $q$-\Col problem on $\ell$-uniform hypergraphs, both parameterized by the number of vertices.
\end{lemma}

\begin{proof}
Fix integers $\ell \geq 2$ and $q \geq 2$, and consider an instance of the $(1,\ell,q)$-\URFC problem, consisting of a graph $G=(V,E)$ on $n$ vertices and a collection $\calF$ of $\ell$-tuples of singletons of vertices from $V$. We describe the desired transformation in two steps. In the first, we transform the instance $(G,\calF)$ into an instance $H$ of the $q$-\Col problem on hypergraphs with edges of size at most $\ell$, and in the second, we adapt it to a hypergraph $H'$ whose edges have size precisely $\ell$.

Let $H=(V,E')$ denote the hypergraph defined by
\[E' = E \cup \{ \cup_{j \in [\ell]}{F_j} \mid (F_1, \ldots, F_\ell) \in \calF\}.\]
Namely, the edge set of $H$ consists of the edges of $G$ as well as the sets of vertices that lie in the singletons of each $\ell$-tuple of $\calF$. Observe that a tuple $(F_1, \ldots, F_\ell) \in \calF$ is uniformly rainbow with respect to a given coloring of $V$ if and only if $\cup_{j \in [\ell]}{F_j}$ is monochromatic. Therefore, the mapping that assigns an input $(G,\calF)$ to its associated hypergraph $H$, which can clearly be implemented in polynomial time, forms a transformation from the $(1,\ell,q)$-\URFC problem into the $q$-\Col problem on hypergraphs with edges of size at most $\ell$.

Next, let $H'$ be the hypergraph obtained from $H$ as follows. Add to $H$ a copy of the complete $\ell$-uniform hypergraph on a set $Z$ of $(\ell-1)q$ vertices, and for each edge $e$ in $H$ of size $|e| < \ell$, remove $e$ from the edge set and replace it with all edges of the form $e \cup S$ for subsets $S \subseteq Z$ of size $|S|=\ell-|e|$. By definition, the hypergraph $H'$ is $\ell$-uniform, and it can clearly be constructed in polynomial time. We claim that $H$ admits a proper coloring with $q$ colors if and only if $H'$ does. Indeed, suppose that there exists a proper coloring $c:V \to [q]$ of $H$, and extend it to $H'$ by assigning each of the $q$ colors to $\ell-1$ vertices in $Z$. Such a coloring of $H'$ is proper, because no $\ell$ vertices of $Z$ receive the same color, and because every edge of $H'$ that includes a vertex of $V$ contains an edge of $H$ and is thus not monochromatic. For the other direction, suppose that there exists a proper coloring of $H'$ with $q$ colors. This coloring must assign each of the $q$ colors to precisely $\ell-1$ vertices among the $(\ell-1)q$ vertices of $Z$, as otherwise, some edge in the complete $\ell$-uniform hypergraph on $Z$ that was added to $H$ would be monochromatic. We claim that the restriction of the given coloring of $H'$ to the vertices of $V$ is a proper coloring of $H$. Indeed, suppose for the sake of contradiction that some edge $e$ of $H$ is monochromatic. Since $\ell-1$ vertices in $Z$ are assigned the color of the vertices of $e$, there exists a (possibly empty) set $S \subseteq Z$ for which the edge $e \cup S$ in $H'$ is monochromatic, yielding a contradiction. This completes the proof.
\end{proof}

We are ready to derive Theorem~\ref{thm:Hyper}.

\begin{proof}[ of Theorem~\ref{thm:Hyper}]
Fix integers $\ell \geq 2$ and $q \geq 3$ and a real $\eps >0$. By Definition~\ref{def:eta}, we have $\eta(1,\ell,q)= \ell$. By Theorem~\ref{thm:Lower-URFC}, applied with $d=1$, the $(1,\ell,q)$-\URFC problem parameterized by the number of vertices $n$ does not admit a compression of size $O(n^{\ell-\eps})$ unless $\NP \subseteq \coNPpoly$. By Lemma~\ref{lemma:URFC<qCol}, there is a linear-parameter transformation from this problem to the $q$-\Col problem on $\ell$-uniform hypergraphs, parameterized by the number of vertices $n$. Therefore, by Lemma~\ref{lemma:composition}, if the latter admits a compression of size $O(n^{\ell-\eps})$, then so does $(1,\ell,q)$-\URFC, implying that $\NP \subseteq \coNPpoly$.
\end{proof}

\subsection{Coloring Graphs Close to Disjoint Union of Cliques}

The next application concerns graph coloring problems parameterized by the vertex-deletion distance to disjoint union of cliques. For positive integers $q$ and $t$, consider the parameterized \mbox{$q$-\Col} problem on $t \lClique+k\mathrm{v}$ graphs. Recall that its input consists of a graph $G=(V,E)$ and a set $X \subseteq V$ of size $k$ such that $G \setminus X \in t \lClique$, and the task is to decide whether $G$ admits a proper coloring with $q$ colors, where $k$ serves as the parameter. Note that given an instance $(G,X)$ of the problem, the partition of $G \setminus X$ into disjoint cliques can be found in polynomial time. We assume that $q \geq 3$, as otherwise the $q$-\Col problem is solvable in polynomial time, and that $t \leq q$, since a clique of size larger than $q$ cannot be properly colored with $q$ colors.

The following lemma, inspired by~\cite{Schalken20}, provides a useful connection between coloring disjoint unions of cliques and uniformly rainbow tuples.

\begin{lemma}\label{lemma:cliques_vs_F}
For positive integers $q$ and $t \leq q$, let $G=(V,E)$ be a graph, and let $X \subseteq V$ be a set such that $G \setminus X \in t\lClique$.
Let $\calF$ be the collection of all tuples $(F_1, \ldots, F_{\ell})$ of $(q-\ell+1)$-subsets of $X$ for an integer $\ell \in [t]$ such that there exists a (not necessarily maximal) clique $\{v_1, \ldots, v_{\ell}\}$ in $G \setminus X$ satisfying $F_i \subseteq N_G(v_i)$ for all $i \in [\ell]$.
Then, for every proper coloring $c:X \to [q]$ of $G[X]$, $c$ can be extended to a proper coloring of $G$ with $q$ colors if and only if no tuple in $\calF$ is uniformly rainbow with respect to $c$.
\end{lemma}

\begin{proof}
Consider a graph $G$, a set $X$, and a collection $\calF$ as in the lemma. Let $c:X \to [q]$ be a proper coloring of $G[X]$.
Suppose first that $c$ can be properly extended to $G$. To show that no tuple in $\calF$ is uniformly rainbow with respect to $c$, consider an integer $\ell \in [t]$, a clique $\{v_1, \ldots, v_{\ell}\}$ in $G \setminus X$, and a tuple $(F_1, \ldots, F_{\ell}) \in \calF$ of $(q-\ell+1)$-subsets of $X$ satisfying $F_i \subseteq N_G(v_i)$ for all $i \in [\ell]$. Suppose, towards a contradiction, that $c$ assigns to all sets in the tuple the same set of $q-\ell+1$ distinct colors, which forces any proper extension of $c$ to $G$ to assign to all the vertices of the clique $\{v_1, \ldots, v_{\ell}\}$ colors from the remaining $\ell-1$ colors. Since this clique has $\ell$ vertices with fewer than $\ell$ available colors, $c$ cannot be properly extended to $G$, a contradiction.

Next, suppose that no tuple in $\calF$ is uniformly rainbow with respect to $c$. We extend $c$ to a proper coloring of $G$ as follows. Recall that the graph $G \setminus X$ is a disjoint union of cliques, each of size at most $t$. We show that $c$ can be properly extended to each of them. Consider a maximal clique $C$ in $G \setminus X$, and note that $|C| \leq t$. Let $H$ denote the bipartite graph that has the vertices of $C$ on one side and the color set $[q]$ on the other. We connect each vertex $v \in C$ to the vertices representing its available colors, i.e., those of $[q] \setminus \{c(u) \mid u \in N_G(v) \cap X\}$. Our goal is to show that $H$ admits a matching that saturates the vertices of $C$, allowing us to assign a distinct available color to each vertex in $C$. By Hall's theorem, it suffices to verify that for every set $S \subseteq C$, the total number of neighbors in $H$ of the vertices of $S$ is at least $|S|$. Suppose for contradiction that there exists a set $S = \{v_1, \ldots, v_{\ell}\} \subseteq C$ for some $\ell \in [t]$ that does not satisfy this property, meaning that the total number of neighbors in $H$ of the vertices of $S$ is at most $\ell-1$. It follows that there are $q-\ell+1$ colors that are not available to each of the vertices in $S$, and thus, for each $i \in [\ell]$, there exists a $(q-\ell+1)$-subset $F_i$ of $X$ with $F_i \subseteq N_G(v_i)$ whose vertices are assigned those $q-\ell+1$ colors by $c$. This implies that the tuple $(F_1, \ldots, F_{\ell})$ lies in $\calF$ and is uniformly rainbow with respect to $c$, contradicting our assumption and completing the proof.
\end{proof}

The next lemma relates the $q$-\Col problem on $t \lClique+k\mathrm{v}$ graphs to the $(I,q)$-\GURFC problem for a suitable set $I$ (see Definition~\ref{def:GURFC}).

\begin{lemma}\label{lemma:CliqesVsGURFC}
For positive integers $q$ and $t \leq q$, let $I = \{(q-\ell+1,\ell) \mid \ell \in [t]\}$.
Then there exists a linear-parameter transformation from the $(I,q)$-\GURFC problem parameterized by the number of vertices into the $q$-\Col problem on $t \lClique+k\mathrm{v}$ graphs.
\end{lemma}

\begin{proof}
For positive integers $q$ and $t$ with $t \leq q$, consider an instance of $(I,q)$-\GURFC, namely, a graph $G=(V,E)$ and a collection $\calF$ of $\ell$-tuples of $(q-\ell+1)$-subsets of $V$ for integers $\ell \in [t]$. We define an algorithm that, given such an input, constructs a graph $G'$ obtained from $G$ by adding, for each $\ell$-tuple $(F_1, \ldots, F_{\ell}) \in \calF$ of $(q-\ell+1)$-subsets of $V$, a clique $\{v_1, \ldots, v_{\ell}\}$ of size $\ell$ and connecting each $v_i$ to all vertices in $F_i$. The algorithm returns the graph $G'$ along with the set $X=V$.

The above algorithm can clearly be implemented in polynomial time. The output $(G',X)$ is valid, because $G' \setminus X$ is a disjoint union of cliques of size at most $t$, and the transformation is linear-parameter, because the number of vertices in $X$ is $|V|$. For correctness, we shall show that $(G,\calF)$ is a $\YES$ instance of $(I,q)$-\GURFC if and only if $G'$ is a $\YES$ instance of $q$-\Col. To this end, we apply Lemma~\ref{lemma:cliques_vs_F} to the graph $G'$ and the set $X$. Observe that the given collection $\calF$ precisely coincides with the one from the lemma, implying that any proper coloring $c$ of $G$ can be properly extended to $G'$ if and only if no tuple in $\calF$ is uniformly rainbow with respect to $c$. Thus, there exists a proper coloring of $G$ with $q$ colors for which no tuple in $\calF$ is uniformly rainbow if and only if there exists a proper coloring of $G'$ with $q$ colors, as required.
\end{proof}

Equipped with Lemma~\ref{lemma:CliqesVsGURFC}, we derive near-optimal upper and lower bounds on the kernel complexity of the $q$-\Col problem on $t \lClique+k\mathrm{v}$ graphs for all admissible values of $q$ and $t$, confirming Theorem~\ref{thm:t-Cliques} (see Remark~\ref{remark:encoding_disj}).

\begin{theorem}\label{thm:LowerDisjClique_ell}
For positive integers $q \geq 3$ and $t \leq q$, set
\[
r = r(q,t) =
\begin{cases}
q - 1, & t = 1, \\[1mm]
2q - 3, & t = 2 \text{~~~or~~~}t=q=3, \\[1mm]
\displaystyle (q - t + 1)t, & 3 \le t < \frac{q+1}{2}, \\[1mm]
\displaystyle \left\lfloor (q + 1)^2/4 \right\rfloor, & \text{otherwise.}
\end{cases}
\]
The $q$-\Col problem on $t\lClique+k\mathrm{v}$ graphs admits a kernel with $O(k^r)$ vertices, encodable in $O(k^r \cdot \log k)$ bits, while for any real $\eps >0$, it does not admit a compression of size $O(k^{r-\eps})$ unless $\NP \subseteq \coNPpoly$.
\end{theorem}

\begin{proof}
For positive integers $q \geq 3$ and $t \leq q$, define $I=\{(q-\ell+1,\ell) \mid \ell \in [t]\}$ and set \mbox{$r = \max_{\ell \in [t]} {\eta(q-\ell+1,\ell,q)}$}, with $\eta$ as in Definition~\ref{def:eta}. One can verify that this definition of $r$ coincides with the one stated in the theorem. Since $q \geq 3$, we have $\eta(q-\ell+1,\ell,q) \geq 2$ for all $\ell \in [t]$, and thus $r \geq 2$.

We begin with the upper bound.
Consider the algorithm that given an instance of the $q$-\Col problem on $t\lClique+k\mathrm{v}$ graphs, i.e., a graph $G=(V,E)$ and a set $X \subseteq V$ of size $k$ such that $G \setminus X \in t \lClique$, performs the following steps.
\begin{enumerate}
  \item\label{itm:1_alg_cl} Construct the collection $\calF$ associated with $G$ and $X$ in Lemma~\ref{lemma:cliques_vs_F}, namely, the collection of all tuples $(F_1, \ldots, F_{\ell})$ of $(q-\ell+1)$-subsets of $X$ with $\ell \in [t]$ such that there exists a clique $\{v_1, \ldots, v_{\ell}\}$ in $G \setminus X$ satisfying $F_i \subseteq N_G(v_i)$ for all $i \in [\ell]$.
  \item\label{itm:2_alg_cl} Apply to $(G[X],\calF)$ the kernel of the $(I,q)$-\GURFC problem parameterized by the number of vertices $k$, given in Theorem~\ref{thm:GURFC}. This yields an equivalent instance $(G[X],\calF')$ with $|\calF'| \leq O(k^r)$.
  \item\label{itm:3_alg_cl} Construct the graph $G'$ obtained from $G[X]$ by adding, for each $\ell$-tuple $(F_1, \ldots, F_{\ell}) \in \calF'$ of $(q-\ell+1)$-subsets of $X$ with $\ell \in [t]$, a clique $\{v_1, \ldots, v_{\ell}\}$ of size $\ell$ and connecting each vertex $v_i$ to all vertices of $F_i$.
  \item Return the pair $(G',X)$.
\end{enumerate}

We turn to the analysis of the algorithm. Consider an input of the $q$-\Col problem on $t\lClique+k\mathrm{v}$ graphs, i.e., a graph $G=(V,E)$ and a set $X \subseteq V$ with \mbox{$G \setminus X \in t \lClique$} and $|X|=k$. The graph $G[X]$ has $k$ vertices, and the graph $G'$ is obtained from $G[X]$ by adding $|\calF'|$ cliques, each of size at most $t$, whose vertices have a constant number of neighbors in $X$.
It follows that $G' \setminus X \in t \lClique$, hence the pair $(G',X)$ is a valid output of the kernel. It further follows that the number of vertices in $G'$ is $k+O(k^r) = O(k^r)$, as desired. Furthermore, each clique in $G' \setminus X$ can be represented by the lists of the neighbors in $X$ of its vertices, which can all be encoded in $O(\log k)$ bits. Thus, to represent the graph $G'$, it suffices to use $O(k^2)$ bits for the adjacencies in $G'[X]$ and additional $O(k^r \cdot \log k)$ bits for the cliques in $G' \setminus X$. Since $r \geq 2$, this implies that the kernel is encodable in $O(k^r \cdot \log k)$ bits, as claimed.

For correctness, we shall prove that $G$ admits a proper coloring with $q$ colors if and only if $G'$ does. By Lemma~\ref{lemma:cliques_vs_F}, a proper coloring of $G[X]$ with $q$ colors can be properly extended to $G$ if and only if it has no uniformly rainbow tuple in the collection $\calF$ from Item~\ref{itm:1_alg_cl} of the algorithm. Therefore, $G$ admits a proper coloring with $q$ colors if and only if $(G[X],\calF)$ is a $\YES$ instance of the $(I,q)$-\GURFC problem, which by Theorem~\ref{thm:GURFC}, is equivalent to the instance $(G[X],\calF')$, produced in Item~\ref{itm:2_alg_cl} of the algorithm. Now, apply Lemma~\ref{lemma:cliques_vs_F} to the graph $G'$ from Item~\ref{itm:3_alg_cl} of the algorithm and the set $X$, and observe that the collection of tuples defined in the lemma is precisely $\calF'$. This shows that $(G[X],\calF')$ is a $\YES$ instance of the $(I,q)$-\GURFC problem if and only if $G'$ admits a proper coloring with $q$ colors, and we are done.

We turn to the lower bound. By the definition of $r$, there exists an integer $\ell \in [t]$ such that $r = \eta(q-\ell+1,\ell,q)$. By Lemma~\ref{lemma:CliqesVsGURFC}, there exists a linear-parameter transformation from the $(I,q)$-\GURFC problem parameterized by the number of vertices to the $q$-\Col problem on $t \lClique+k\mathrm{v}$ graphs. Therefore, if for some $\eps >0$, the $q$-\Col problem on $t\lClique+k\mathrm{v}$ graphs admits a compression of size $O(k^{r-\eps})$, then by Lemma~\ref{lemma:composition}, the $(I,q)$-\GURFC problem parameterized by the number of vertices $n$ admits a compression of size $O(n^{r-\eps})$. Since $(q-\ell+1,\ell) \in I$, this also yields a compression of the same size for the $(q-\ell+1,\ell,q)$-\URFC problem parameterized by the number of vertices. By Theorem~\ref{thm:Lower-URFC}, using $r = \eta(q-\ell+1,\ell,q)$, this implies that $\NP \subseteq \coNPpoly$, so the proof is complete.
\end{proof}

\begin{remark}\label{remark:encoding_disj}
We note that the logarithmic factor in the encoding size stated in Theorem~\ref{thm:LowerDisjClique_ell} can be avoided whenever $q \geq 4$ and $t \geq 3$. Indeed, under these conditions, it holds that $r = \max_{\ell \in [t]}{(q-\ell+1)\ell}$. This implies that the number of potential cliques in the graph $G'$ constructed in Item~\ref{itm:3_alg_cl} of the algorithm from Theorem~\ref{thm:LowerDisjClique_ell} is bounded by $O(k^r)$. To represent them, it suffices to use a single bit per each potential clique, indicating whether it lies in $G'$ or not.
\end{remark}

Our final corollary addresses the $q$-\Col problem parameterized by the vertex-deletion distance of the input graph to a disjoint union of cliques, with no specific bound on their size. It resolves a question posed in~\cite{Schalken20}.

\begin{corollary}\label{cor:q-Col-Disj}
For an integer $q \geq 3$, set $r$ to be $3$ if $q=3$, and $\lfloor (q + 1)^2/4 \rfloor$ otherwise.
The $q$-\Col problem on $\Clique +k\mathrm{v}$ graphs admits a kernel with $O(k^r)$ vertices, encodable in $O(k^r \cdot \log k)$ bits if $q=3$ and in $O(k^r)$ bits otherwise, while for any real $\eps >0$, it does not admit a compression of size $O(k^{r-\eps})$ unless $\NP \subseteq \coNPpoly$.
\end{corollary}

\begin{proof}
Fix an integer $q \geq 3$, and notice that $r$ fits the definition of $r(q,q)$ from Theorem~\ref{thm:LowerDisjClique_ell}.
For the upper bound, consider the algorithm that, given a graph $G=(V,E)$ and a set $X \subseteq V$ such that \mbox{$G \setminus X \in \Clique$}, if $G \setminus X$ contains a clique of size larger than $q$, concludes that $G$ has no proper coloring with $q$ colors and returns an arbitrary $\NO$ instance of fixed size, and otherwise, applies the algorithm from Theorem~\ref{thm:LowerDisjClique_ell} for $t=q$. This gives a kernel with $O(k^r)$ vertices, that can be encoded, by Remark~\ref{remark:encoding_disj}, in $O(k^r \cdot \log k)$ bits if $q=3$ and in $O(k^r)$ bits otherwise. The lower bound follows from Theorem~\ref{thm:LowerDisjClique_ell} applied with $t=q$, using the obvious fact that any graph in $q\lClique$ also belongs to $\Clique$.
\end{proof}

\section{Concluding Remarks}\label{sec:conclude}

This paper studies the kernel complexity of $R$-\CC problems, a class of constraint satisfaction problems that combine the non-inequality relation over a finite domain together with an additional relation $R$. We provide a general conditional lower bound on the compression size of such problems when parameterized by the number of variables and demonstrate its applicability to several concrete parameterized coloring problems. These include the $(d,\ell,q)$-\URFC problems and the $q$-\Col problems on $\ell$-uniform hypergraphs, both parameterized by the number of vertices, as well as the $q$-\Col problems on graphs, when parameterized by the vertex-deletion distance to a disjoint union of cliques of size at most $t$. For all of these problems, and for all admissible values of their relevant constants, we obtain nearly matching upper and lower bounds on the kernel complexity under the assumption $\NP \nsubseteq \coNPpoly$.

It will be interesting to find additional applications of the bounds presented here for determining the kernel complexity of other coloring problems. Natural candidates are the families of $H$-\Col and \ListHCol problems, which concern the existence of a homomorphism from a given graph to a fixed target graph $H$, when parameterized by the number of vertices, by the vertex cover number, or by other structural parameters of interest (see, e.g.,~\cite{JansenP19color,ChenJOPR23,BerkmanH25,PPRW25}). A more ambitious objective would be to characterize, under a plausible complexity assumption, the behavior of the kernel complexity of the $R$-\CC problem for an arbitrary relation $R$.

\bibliographystyle{abbrv}
\bibliography{rainbow}

\end{document}